\documentclass[psamsfonts]{article}
\usepackage{amsmath, amsthm, amssymb}
\usepackage[dvips]{graphics}
\usepackage[]{color}
\usepackage{subfigure}
\usepackage{multirow}
\usepackage{graphicx} 
\usepackage{bm}     
\usepackage{setspace} 
\usepackage{booktabs} 
\usepackage{verbatim}
\usepackage{appendix} 
\usepackage{fullpage} 

\numberwithin{equation}{section}

\theoremstyle{plain}
\newtheorem{thm}{Theorem}[section]

\newtheorem{corr}[thm]{Corollary}
\newtheorem{lemma}[thm]{Lemma}

\theoremstyle{definition}

\theoremstyle{remark}

\newcommand{\un}{\underline}

\newcommand{\mr}{\mathrm}
\newcommand{\mc}{\mathcal}
\newcommand{\ra}{\rightarrow}

\newcommand{\lra}{\longrightarrow}

\newcommand{\mv}{\mathbf}

\renewcommand{\t}{\text}
\renewcommand{\it}{\textit}

\newcommand{\Sig}{\Sigma}
\newcommand{\sig}{\sigma}
\newcommand{\ep}{\epsilon}

\newcommand{\del}{\delta}
\newcommand{\alp}{\alpha}

\newcommand{\til}{\tilde}

\renewcommand{\th}{\theta}
\newcommand{\Th}{\Theta}
\renewcommand{\l}{\lambda}

\newcommand{\s}{\sqrt}

\begin{document}



\title{Density Estimation and Classification via Bayesian Nonparametric Learning of Affine Subspaces}
\author{ Abhishek Bhattacharya \\ Indian Statistical Institute \\ Kolkata India \\   abhishek@isical.ac.in
	\and	Garritt  Page \\ Department of Statistical Science \\ Duke University \\ page@stat.duke.edu
     	\and David  Dunson \\ Department of Statistical Science \\ Duke University \\ dunson@stat.duke.edu}

\maketitle

\begin{abstract}

It is now practically the norm for data to be very high dimensional in areas such as genetics, machine vision, image analysis and many others. When analyzing such data, parametric models are often too inflexible while nonparametric procedures tend to be non-robust because of insufficient data on these high dimensional spaces.  It is often the case with high-dimensional data that most of the variability tends to be along a few directions, or more generally along a much smaller dimensional submanifold of the data space.      In this article, we propose a class of models that flexibly learn about this submanifold and its dimension which simultaneously performs dimension reduction.  As a result, density estimation is carried out efficiently.  When performing classification with a large predictor space, our approach allows the category probabilities to vary nonparametrically with a few features expressed as linear combinations of the predictors.  As opposed to many black-box methods for dimensionality reduction, the proposed model is appealing in having clearly interpretable and identifiable parameters.  Gibbs sampling methods are developed for posterior computation, and the methods are illustrated in simulated and real data applications.

\end{abstract}

{{\bf keywords}: Dimension reduction;  Classifier; Variable selection;  Nonparametric Bayes}


\doublespace
\section{Introduction} \label{s0}
Data that are generated from experiments or studies carried out in areas such as genetics, machine vision, and image analysis (to name a few) are routinely high dimensional.  Because such data sets have become so commonplace,  designing data efficient inference techniques that scale to massive dimensional Euclidean and even non-Euclidean spaces has attracted considerable attention in the statistical and machine learning literature. 

When dealing with high dimensional data, it is typically the case that parametric models are too rigid to explain all the variability present in the data.  Conversely,  flexible nonparametric approaches suffer from the well known curse of dimensionality.  With this in mind, a common approach is to make procedures more scalable to high dimensions by learning a lower dimensional subspace the data are concentrated near.  This approach is supported by the success of mixture models with a few components in fitting high-dimensional data.  In particular,  consider a mixture of $N$ Gaussian kernels, $\sum_{j=1}^N \pi_jN_m(\cdot;\mu_j,\sig^2I_m)$, $\mu_j\in \Re^m$.   The $k=N-1$ largest eigenvalues corresponding to the covariance matrix for this type of density will typically be very large,  while the remaining $m-k$ eigenvalues will all be equal and relatively much smaller. We may visualize such data lying close to some affine $k$ dimensional subspace of $\Re^m$ containing the mean and the $k$ corresponding eigen-vectors as its directions. If we knew that subspace, we could model the data projected onto that subspace with a nonparametric density model, while using some simple parametric distribution on the orthogonal residual vector. Robustness would be attained by fitting a flexible model on only a selected few coordinates.


There is a large literature on the estimation of Euclidean subspaces, affine subspaces, and manifold subsets.  Many procedures are algorithmic based.     Elhamifar  and  Vidal \cite{elhamifar} propose an algorithmic based method of clustering data that lie close to multiple affine subspaces.   See the references there in for a nice overview of algorithmic type approaches. Because such methods are deterministic, no measures of uncertainty are available.  A probabilistic modeling approach is proposed by Chen \emph{et al.}~\cite{chen}.  They employ a fully Bayesian model for density estimation of high dimensional data that reside close to a lower dimensional subregion (possibly a manifold) of unknown dimension.  This subregion is approximated using a nonparametric Bayes mixture of factor analyzers in which Dirichlet and beta processes are employed to simultaneously allow uncertainty in the number of mixture components, the number of factors in each component and the locations of zeros in the loadings matrix.   Although their methodology is flexible, it is very much a complex and over-parametrized ``black box'' leading to challenging computation.   

We propose a fully Bayesian procedure that very flexibly and uniquely identifies a lower dimensional affine subspace in a coherent modeling framework.  After having identified the subspace and its dimension we model the coordinates of the orthogonal projection of the data onto that subspace using an infinite mixture of Gaussians while independently using a zero mean Gaussian to model the data component orthogonal to that subspace. Among all possible coordinate choices, we prefer isometric coordinates  (those which preserve the geometry of the space). To obtain such coordinates, an orthogonal basis for the subspace must be employed which will require working on the Stiefel manifold  (the space of all such basis matrices). In addition to interpretability and identifiability, advantages to using an orthogonal basis include equivalence of matrix inversion and transpose and faster MCMC convergence.  We do not limit the cluster contours to be homogeneous, but use a singular value decomposition type sparse representation for the kernel covariance. By doing so, we avert the problem of dealing with massive matrices and yet make the model highly flexible.

An appealing feature to our methodology is that it is not a ``black box", rather nice interpretations accompany model parameters. For example, when estimating the affine subspace, which is proved to be unique, concern lies in estimating the orthogonal projection matrix associated with that space, and its orthogonal shift from the origin. Indeed, under our setting, the subspace turns out to be the $k$-principal subspace for the distribution, $k$ being the subspace dimension. In this regard, the methodology developed here provides a coherent extension of the Principal Component Analysis (PCA) of Hoff \cite{hoff1} to a nonparametric setting. The estimation of the projection matrix and orthogonal shift are carried out explicitly under appropriate loss functions.


We also consider building efficient classifiers that entertain a high dimensional feature space. The idea is to seek the minimal subspace of the feature space such that the response depends on the predictors only through their projection onto that subspace.  There has been recent developments in the machine learning and statistical communities with regards to building classifiers in the presence of a high dimensional feature space.  Sun \emph{et al.}~\cite{YijunSun} propose a classifier that essentially breaks a complex nonlinear problem into a set of local linear problems that scales nicely to a very high dimensional space.  They also provide a nice review of algorithmic based procedures to building classifiers most of which are black boxes and estimation of a principal subspace is not entertained.  Recently, Cucala \emph{et al.}~\cite{cucala} proposed a probabilistic perspective to the $k$-nearest neighbor classifiers.  However, apart from not scaling well to a high dimensional feature space, the minimal subspace of the feature space is not estimated.    Estimating a minimal subspace of a high dimensional feature space has been addressed in a regression setting.   Tokdar \emph{et al.}~\cite{tokdar} model the conditional distribution of a response given the minimal subspace directly with a Gaussian process.  Recently, Reich \emph{et al.}~\cite{reich} propose a method of sufficient dimension reduction by modeling a conditional distribution directly after placing a prior distribution on the minimal subspace (which they call a central subspace).   See references there in for frequentist approaches to estimating this subspace.    Hannah \emph{et al.}~\cite{hannah} use Dirichlet process mixtures to flexibly model the relationship between a set of features and a response in a generalized linear model framework.   Shahbaba and Neal \cite{shahbaba} focus on Dirichlet process mixture models in a nonlinear modeling framework.   


We focus on modeling the joint so that given the  subspace, the response and the projection of the features onto that subspace follow a nonparametric infinite mixture model while the feature component orthogonal to the subspace follows a parametric model independent of the response and the projection. Dependence between the response and features is induced through the mixture distribution.  


The remainder of this article is organized as follows.  Section 2 provides some preliminaries,  Section 3 details the class of models to be used for density estimation along with 
theoretical results dealing with large prior support and strong posterior consistency.  In Section 4 we investigate the identifiability of model parameters and give details of their estimation.   
Section 5 details computational strategies while Section 6 outlines a small simulation study and examples.   In Section 7 we develop an efficient classifier and provide some examples and a small simulation study in addition to briefly introducing ideas with regards to regression.  We finish with some concluding remarks in Section 8.

\section{Preliminaries}
A $k$-dimensional affine subspace of $\Re^m$ (which is a $k$-dimensional Euclidean manifold) can be expressed as
$$
S = \{ Ry + \th \colon y \in \Re^m \}
$$
with $R$ being a $m \times m$ rank $k$ \it{projection matrix} (it satisfies
$R = R' = R^2$, rank($R) = k$)  and $\th \in \Re^m$ satisfying $R\th = 0$.
Notice that  there is a one to one correspondence between the subspace $S$ and the pair $(R,\th)$ 
with $\th$ being the \it{projection} of the origin into $S$ and $R$ the projection matrix of the shifted linear subspace 
$$ L = S - \th = \{ Ry \colon y \in \Re^m \}.$$ 
The projection of any $x\in \Re^m$ into $S$ is defined as the $x_0 \in S$ satisfying $\|x - x_0\| = \min\{ \|x - y \|: y \in S \}$ where $\|\cdot\|$ denotes the Euclidean norm.  For any affine subspace $S$ as defined above, the solution turns out to be $x_0 = Rx + \th$. Similarly, the projection of $x \in \Re^m$ into $L$ is $x_0^* =  Rx$, hence the name projection matrix for $R$. We denote the projection of  $x\in \Re^m$ into $S$ as $Pr_S(x)$. 


Each $x \in \Re^m$ can be given coordinates $\til x \in \Re^k$ such that $x = U\til x + \th$ where $U$ is a matrix  whose columns $\{U_1,\ldots, U_k\}$ form a basis of the column space of $R$.   If $U$ is chosen to be orthonormal (i.e., $U'U = I_k$ and $R=UU'$), then the coordinates ($\tilde{x}$) are \emph{isometric}.  That is, they preserve the inner product on $S$ (and hence volume and distances).   With such a basis, the projection $Pr_S(x)$ of an arbitrary $x \in \Re^m$ into $S$ has isometric coordinates $U'x$.
Thus, $U$ gives $k$ mutually perpendicular `directions' to $S$ while  $\th$ may be viewed as the `origin' of $S$.  We will call $\th$ the \it{origin} and $U$ an \it{orientation} for $S$.


The \it{residual} of $x\in \Re^m$ (which we denote as $R_S(x) = x - Pr_S(x) = x - Rx - \th$)  lies on a linear subspace that is perpendicular to $L$.  That is, $R_S(x) \in S^\perp$ where
$$S^\perp = \{ (I-R)y \colon y \in \Re^m \}. $$
Notice that the projection matrix of $S^\perp$ is $I-R$.  Now if we let $V$ denote an orthonormal basis for the column space of $I-R$ (i.e., $V'V=I_{m-k}$, $VV' = I-R$), then isometric residual coordinates are given by  $V'x \in \Re^{m-k}$.

For a sample lying close to such a subspace $S$, it is natural to assume that the data residuals are centered around $0$ with low variability while the data projected into $S$ comes from a possibly multi-modal distribution supported on $S$.  Figure \ref{2Dpic} illustrates such a sample cloud. The observations are drawn from a two-component mixture of bivariate normals with cluster centers $(1,0)$ and $(0,1)$  and band-width of 0.5. As a result they are clustered around the subspace (line) $x+y =1$. For a specific sample point $x$, $Pr_S(x)$, $R_S(x)$, and $\th$ are highlighted.
\begin{figure}[htbp]\label{fig1}
\begin{center}
\includegraphics[scale=0.6]{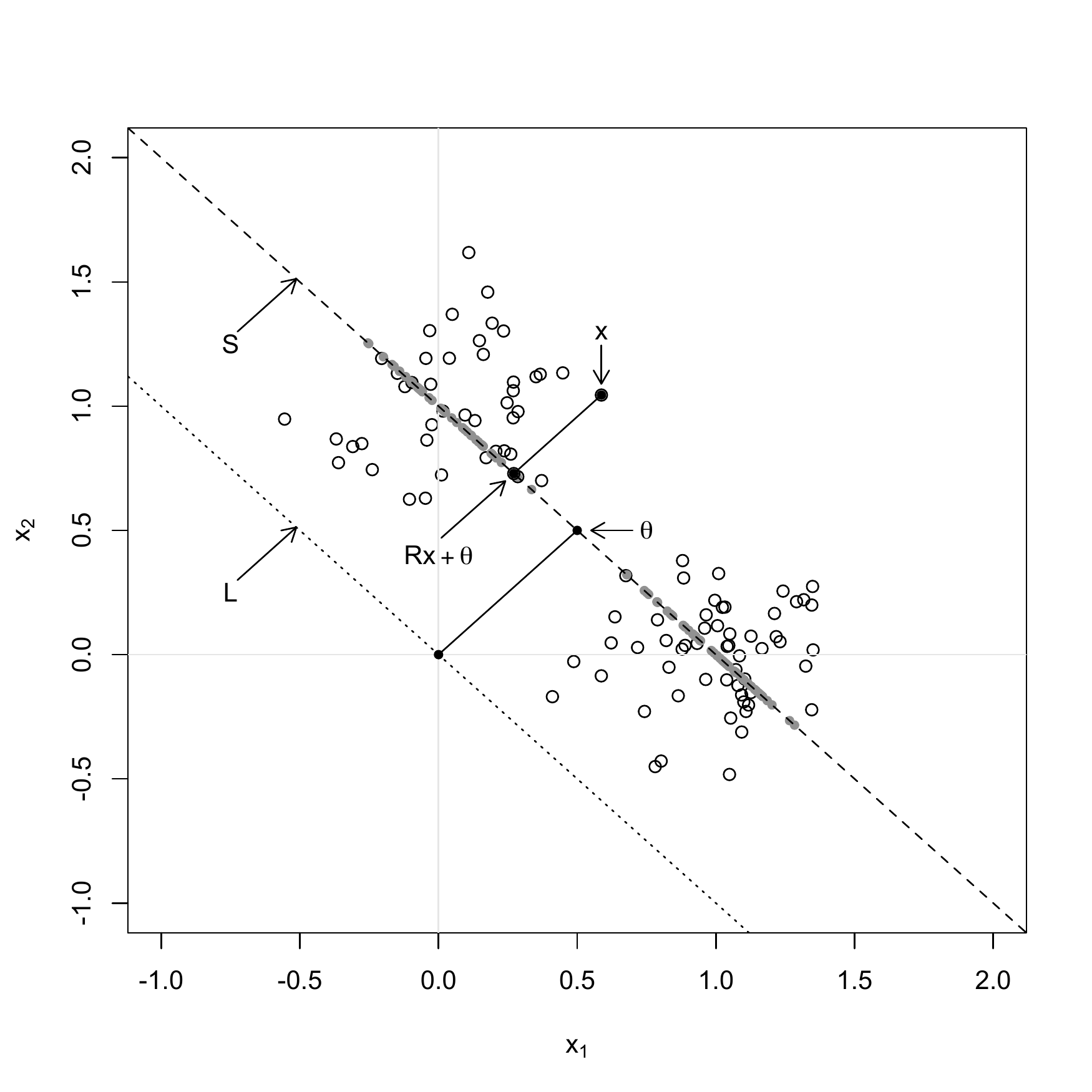}
\caption{Graphical representation of the affine subspace ($S$), the orthogonal shift ($\theta$), and the projection of a point into $S$ (these are the solid dots with particular emphasis given to $Rx + \theta$). }
\label{2Dpic}
\end{center}
\end{figure}

If  we let $Q$ to be a distribution on $\Re^m$ with finite second order moments, then for $d \le m$ the $d$ \it{principal affine subspace} of $Q$ is the minimizer of following risk function
\begin{eqnarray}\label{e1}
R(S) = \int_{\Re^m} \| x - Pr_S(x) \|^2 Q(dx),
\end{eqnarray}
with the minimization carried out over all $d$-dimensional affine subspaces $S$. The minimum value of expression 2.1 turns out to be $\sum_{d+1}^m \l_j$, where $\l_1 \ge \ldots \ge \l_m$ are the ordered eigenvalues of the covariance of $Q$. In addition,  a unique minimizer exists if and only if $\l_d > \l_{d+1}$. If this is indeed the case, then the $d$ principal affine subspace ($S_o$) has projection matrix $R = UU'$ (here $U$ is any orthonormal basis for the subspace spanned by a set of $d$ independent eigenvectors corresponding to the first $d$ eigenvalues) and origin $\th = (I-R)\mu$ (with $\mu$ being the mean of $Q$). Notice that when $d=0$, $S_o$ is the point set $\mu$.

In the case that $d$ is unknown,  we can find an optimal value of $d$ by considering 
\begin{eqnarray}\label{e3}
R(d,S) = f(d) + \int_{\Re^m} \| x - Pr_S(x) \|^2 Q(dx), \ 0 \le d \le m
\end{eqnarray}
as a risk function for some fixed increasing convex function $f$. For $f$ linear, say, $f(d) = ad$, $a> 0$, the risk has a unique minimizer if and only if $\l_{d+1} < a < \l_d$ for some $d$, with $\l_0 = \infty$ and $\l_{m+1} = 0$. Then the minimizing dimension $d_o$ is that value of $d$ while the optimal space $S_o$ is the $d_o$ principal affine subspace.
We will call $d_o$ the \it{principal dimension} of $Q$. For the observations in Figure \ref{2Dpic}, the principal dimension is $d_o = 1$ with principal subspace
\begin{align*}
S_o = \left\{ \left(\begin{array}{cc} 1/2   &  -1/2  \\ -1/2  &  1/2  \\ \end{array} \right)x + 
         \left(\begin{array}{c} 1/2 \\  1/2  \\ \end{array} \right)
   :x \in \Re^2 \right\}.
\end{align*}
 
Before detailing general modeling strategies, we introduce notation that will be used through out. By $\mc{M}(S)$ we denote the space of all probabilities on the space $S$.  $M(m,k)$ will denote real matrices of order $m\times k$ (with $M(m)$ denoting the special case of $m=k$), $M^+(m)$ will denote the space of all $m\times m$ positive definite matrices. For $U\in M(m,k)$, $\mc{C}(U)$ and $\mc{N}(U)$ will represent the column and null space of $U$ respectively.  We will represent the space of all $m\times m$ rank $k$ projection matrices by $P_{k,m}$. That is,  
\[
P_{k,m} = \{ R \in M(m)\colon R=R'=R^2, \t{rank}(R)=k \}.
\]

One important manifold referred to in this paper is the Steifel manifold (denoted by $V_{k,m}$) which is the space whose points are $k$-frames in $\Re^m$ (here $k$-frame refers to a set of $k$ orthonormal vectors in $\Re^m$).   That is,
\[
V_{k,m} = \{A \in M(m,k): A'A = I_k\}.
\] 
We denote the orthogonal group $\{A\in\Re^m\colon A'A = I_m\}$ by $O(m)$ which is $V_{m,m}$. The space $V_{k,m}$  is a compact non-Euclidean Riemannian manifold. Because $M(m,k)$ is embedded in Euclidean space, 
it inherits the Riemannian metric tensor which can be used to define the volume form, which in turn can be used as the base measure to construct a parametric family of densities. Several parametric densities have been studied on this space, and exact or MCMC sampling procedures exist. For details, see Chikuse \cite{chikuse2}.
One important density which we will be using as a prior is the Bingham-von Mises-Fisher density which has the expression
\begin{align*}
BMF(x;A,B,C) \propto \mr{etr}(A' x + C x' B x).
\end{align*}
The parameters are $A \in M(k,m)$, $B \in M(k)$ symmetric and $C \in M(m)$, while etr denotes exponential trace.  As a special case, we obtain the uniform distribution which has the constant density $1/\t{Vol}(V_{k,m})$.

\section{Density model }\label{s2}
Consider a random variable $X$ in $\Re^m$. Let there be a $k$ dimensional affine subspace $S$, $0\le k \le m$, with projection matrix $R$ and origin $\th$ such that the projection of $X$ into this subspace follows a location mixture density on the subspace (with respect to its volume form) given by
$$
Y = Pr_S (X) \sim \int_S (2\pi)^{-k/2} |U'AU|^{1/2} \exp \{-\frac{1}{2} (y - w)'A(y-w)\} Q(dw) \\
$$  
where $y \in S$ is the projection of $x$ with parameters $Q \in \mc{M}(S)$,  $U \in V_{k,m}$, and $A$  a $m\times m$ positive semi-definite (p.s.d.) matrix such that $U'AU \in M^+(k)$.  When $k=0$, $S$ denotes the point set $\{\th\}$ and $Y=\th$.  Note that the density expression depends on $U$ only through $UU'$.   A general choice for $A$ besides being positive definite (p.d.) could be $A = U_0 \Sigma^{-1}_0 U_0'$ for some specific orientation $U_0$ and p.d. $\Sigma_0 \in M^+(k)$.
As a result, the isometric coordinates $U_0'X$ of $Pr_S (X)$ follow a non-parametric Gaussian mixture model on $\Re^k$ given by
\begin{eqnarray}\label{e4}
U_0'X \sim \int_{\Re^k} N_k(\cdot;\mu,\Sigma_0)P(d\mu),\ P \in \mc{M}(\Re^k).
\end{eqnarray}
Here $\mu = U'_0w$ for $w \in S$.
Independently, let the residual $R_S(X)$ follow a mean zero homogeneous density on  $S^{\perp}$ given by
$$R_S(X) \sim \sig^{-(m-k)}\exp\{-\frac{\|x\|^2}{2\sig^2}\},$$
$x\in S^{\perp}$ and parameter $\sig > 0$.  If $k=m$, then $S^{\perp} = \{0\}$ and $R_S(X) = 0$.
As a result, with any orientation $V \in V_{m-k,m}$ for $S^{\perp}$, the isometric coordinates $V'X$ of $R_S(X)$ follow the Gaussian density
\begin{equation}\label{e5}
V'X \sim N_{m-k} (;V'\th,\sig^2I_{m-k}).
\end{equation}
Combine equations \eqref{e4} and \eqref{e5} to get the full density of $X$ as
\begin{align}
X \sim f(x;\Th) & = \int_{\Re^k} N_m(x; \phi(\mu), \Sig  )P(d\mu), \label{e7}\\
\phi(\mu) = U_0\mu + \th, \ \Sig & = U_0(\Sigma_0 - \sig^2 I_k)U_0' + \sig^2 I_m, \label{e90}
\end{align}
with parameters $\Th = (k,U_0,\th, \Sigma_0, \sig, P)$. Here $U_0 \in V_{k,m}$ and $\th\in\Re^m$ satisfies $U_0'\th = 0$.
The affine subspace $S$ has projection matrix $R = U_0 U_0'$ and origin $\th$.  For $k=0$, $f(x;\Th) = N_m(x; \th, \sig^2I_m)$.  Using a flexible multimodal density model for a few data coordinates (which are chosen using a suitable basis) and an independent centered Gaussian structure on the remaining coordinates allows efficient density estimation on very high dimensional spaces.  

A common choice of nonparametric prior on $P$ can be a full support discrete model, such as a Dirichlet process, which allows clustering of the data around $S$.  
An alternative way to identify the intercept $\th$ would be to set it equal to $\mr{E}(X)$.   However, this would require the prior on $P$ to be such that $\bar\mu \equiv \int \mu P(d\mu)=0$ making the Dirichlet process prior inappropriate. For this reason, we set $\theta$ to be the origin of $S$ instead.

With $\Sigma_0$ p.d. and $\sig^2 > 0$, the within cluster covariance $\Sig$ lies in $M^+(m)$ and has a sparse representation without being homogeneous. The residual variance $\sig^2$ dictates how ``close" $X$ lies to $S$, with $\sigma^2=0$ implying that $X \in S$.  In \eqref{e7}, one may mix across $\Sigma_0$ by replacing $P(d\mu)$ by $P(d\mu \ d\Sigma_0)$ and achieve more generality. 

To make model \eqref{e7} even more sparse, without loss of generality, we can allow $\Sigma_0$ to be a p.d. diagonal matrix. To prove that we do not lose any generality, consider a singular value decomposition (s.v.d.) of a general $\Sigma_0$, say $\Sigma_0 = ODO'$, $O \in O(k)$, and  replace $\Sigma_0$ by diagonal $D$, and $U_0$ by $U_0O'$.  If $P$ is appropriately transformed, then the model is unaffected.  With a diagonal $\Sigma_0$, the within cluster covariance has $k$ eigenvalues from $\Sigma_0$ and the rest all equal to $\sig^2$. The columns of $U_0$ are the orthonormal eigenvectors corresponding to $\Sigma_0$.

It is easy to check that $S$ is the $k$-principal subspace for the model, if and only if $\Sigma_0 + \int_{\Re^k} (\mu - \bar\mu)(\mu - \bar\mu)'P(d\mu) > \sig^2 I_k$. Here $A_1 > A_2$ refers to $A-B$ being p.d. This holds, for example, when $\Sigma_0 \ge \sig^2 I_k$ and $P$ is non-degenerate. Further under the model, $k$ is the principal dimension of $X$ for a range of risk functions as in \eqref{e3} with linear $f$.

\subsection{Weak Posterior Consistency}\label{s2.1}
Consider a mixture density model $f$ as in \eqref{e7}.
Let $\mc{D}(\Re^m)$ denote the space of all densities on $\Re^m$.  Let $\Pi_f$ denote the prior induced on $\mc{D}(\Re^m)$ through the model and suitable priors on the parameters.  Theorem \ref{t1} shows that  $\Pi_f$ satisfies the Kullback-Leibler (KL) condition at the true density $f_t$ on $\Re^m$.
That is, for any $\ep > 0$, $\Pi_f(K_\ep(f_t)) > 0$, where $K_\ep(f_t) = \{ f \colon KL(f_t;f) < \ep \}$ denotes a $\ep$-sized KL neighborhood of $f_t$ and $KL(f_t;f) = \int\log\frac{f_t}{f}f_t dx$ is the KL divergence.  
As a result, using the Schwartz theorem \cite{schwartz}, weak posterior consistency follows. That is, given a random sample $\mv{X}_n=$ $X_1,\ldots,X_n$ i.i.d. $f_t$,  the posterior probability of any weak open neighborhood of $f_t$ converges to 1 a.s. $f_t$.

Let $p(k)$ denote the prior distribution of $k$.  We consider discrete priors that are supported on the set $\{0,\ldots,m\}$. 
Let $\pi_1(U_0,\th|k)$ denote some joint prior distribution of $U_0$ and $\th$ that has support on $\{(U_0,\th)\in V_{k,m}\times\Re^m: U_0'\th = 0 \}$.  
As previously recommended, we consider a diagonal $\Sig_0 = diag(\sig_1^2, \ldots, \sig_k^2)$ and set a joint prior on the vector $\bm{\sig} = (\sig,\sig_1,\ldots,\sig_k) \in (\Re^+)^{k+1}$ that we denote with $\pi_2(\bm{\sig}|k)$.  Further, we assume that parameters ($U_0$,  $\th$), $\bm{\sig}$, and $P$ are jointly independent given $k$.  That said, Theorem \ref{t1} can be easily adapted to other prior choices.   We also consider the following reasonable conditions on the true density $f_t$. 
\begin{itemize}
\item[{\bf A1}:] $0 < f_t(x) < A$ for some constant $A$ for all $x\in \Re^m$.
\item[{\bf A2}:]  $|\int \log\{f_t(x)\} f_t(x)dx| < \infty$.
\item[{\bf A3}:] For some $\del > 0$, $\int \log\frac{f_t(x)}{f_{\del}(x)}f_t(x)dx < \infty$, where $f_{\del}(x) = \mathop{\inf}_{y:\|y-x\|<\del} f_t(y)$.
\item[{\bf A4}:] For some $\alp > 0$, $\int \|x\|^{2(1+\alp)m}f_t(x)dx < \infty$.
\end{itemize}


\begin{thm}\label{t1}
Set the prior distributions for $k$, ($U_0$,  $\th$), $\bm{\sig}$, and $P$ to those described previously such that $p(m)>0$, $\pi_2(\Re^+\times (0,\ep)^m  | k=m)>0$ for any $\ep>0$, and the conditional prior on $P$ given $k=m$ contains $P_{f_t}$ in its  weak support.  Then under assumptions {\bf A1}-{\bf A4} on $f_t$, the KL condition is satisfied by $\Pi_f$ at $f_t$.
\end{thm}

\begin{proof}
The result follows if it can be proved that $\Pi_f(K_\ep(f_t)| k=m,U_0)>0$ for all $\ep>0$ and $U_0 \in O(m)$, because then
\begin{eqnarray*}
\Pi_f(K_\ep(f_t)) \ge p(m)\int_{O(m)} \Pi_f(K_\ep(f_t)| k=m,U_0)d\pi_1(U_0|k=m) > 0
\end{eqnarray*}
Now, given $k=m$ and $U_0$, density \eqref{e7} can be expressed as
\begin{eqnarray}\label{e2}
f(x; Q, \Sig) = \int_{\Re^m} N_m(x; \nu , \Sig) Q(d\nu),
\end{eqnarray}
with $Q = P\circ \phi^{-1}$.  Here $\phi(x) = U_0x$, and $\Sig = U_0 \Sig_0 U_0'$.
The isomorphism $\phi: \Re^m \ra \Re^m$ being continuous and surjective ensures the same for the mapping $P \mapsto Q$. This in turn ensures that under the Theorem assumptions on the prior, the prior on $P$ and $\bm{\sig}$ induces a prior on $Q$ that contains $P_{f_t}$ in its weak support and an independent prior on $\Sig$  which induces a prior on its maximum eigen-value that contains $0$ in its support.   Then with a slight modification to the proof of Theorem 2 in Wu and Ghosal \cite{WuGhosal2010}, under assumptions {\bf A1-A4} on $f_t$,  we can show that $f_t$ is in the KL support of $\Pi_f$.
\end{proof}

\subsection{Strong Posterior Consistency}\label{s2.2}
Using the density model \eqref{e7} for $f_t$, Theorem \ref{t4} establishes strong posterior consistency, that is, the posterior probability of any total variation (or $L_1$ or strong) neighborhood of $f_t$ converges to 1 almost surely or in probability,  as the sample size tends to infinity.
The priors on the parameters are chosen as in Section \ref{s2.1}. To be more specific, the conditional prior on $P$ given $k$ ($k\ge 1$) is chosen to be a Dirichlet process $DP(w_k P_k)$ ($w_k>0$, $P_k \in \mc{M}(\Re^k)$). 
The proof requires the following three Lemmas.  The proof of Lemma \eqref{l3} can be found in \cite{barron}, while the proofs of Lemmas \eqref{l1} and \eqref{l2} are provided in the appendix.   

In what follows $B_{r,m}$ refers to the set $\{ x \in \Re^m \colon \|x\| \le r \}$. For a subset $\mc{D}$ of densities and $\ep>0$, the $L_1$-metric entropy $N(\ep,\mc{D})$ is defined as the logarithm of the minimum number of $\epsilon$-sized (or smaller) $L_1$ subsets needed to cover $\mc{D}$.

\begin{lemma}\label{l3}
Suppose that $f_t$ is in the KL support of the prior $\Pi_f$  on the density space $\mc{D}(\Re^m)$.
For every $\ep>0$, if we can partition $\mc{D}(\Re^m)$ as $\mc{D}_n^{\ep} \cup \mc{D}_n^{\ep c}$ such that $N(\ep,\mc{D}_n^{\ep})/n \lra 0$ and $Pr(D_n^{\ep c}| \mv{X}_n) \lra 0$ a.s. or in probability $P_{f_t}$, 
then the posterior probability of any 
$L_1$ neighborhood of $f_t$ converges to 1 a.s. or in probability $P_{f_t}$.
\end{lemma}

\begin{lemma}\label{l1}
For positive sequences $h_n \ra 0$ and $r_n\ra \infty$ and $\ep>0$, define a sequence of subsets of $\mc{D}(\Re^m)$ as
$$ \mc{D}_n^\ep = \{ f(\cdot;\Th): \Th \in H_n^\ep \}, \ H_n^\ep = \{ \Th \colon \min(\bm{\sig}) \ge h_n, \|\th\| \le r_n, P(B_{r_n,k}^c) < \ep \} $$
with $f(\cdot;\Th)$ as in \eqref{e7}.  Set a prior on the density parameters as in Section \ref{s2.1}. Assume that $supp(\pi_2(\cdot|k)) \subseteq [0,A]^{k+1}$ for some $A>0$ for all $0\le k \le m$.
Then  $N(\ep,\mc{D}_n^\ep) \le C (r_n/h_n)^m$ where $C$ is a constant independent of $n$.
\end{lemma}

\begin{lemma}\label{l2}
Set a prior as in Lemma~\ref{l1} with a $DP(w_kP_k)$ prior on $P$ given $k$, $k\ge 1$.
Assume that the base probability $P_k$ has a density $p_k$ which is positive and continuous on $\Re^k$.
Assume that there exist positive sequences $h_n \ra 0$ and $r_n\ra \infty$ such that 
$$
{\bf B1}: \lim_{n\ra\infty}n\del_{kn}^{-1}h_n^{-k}\exp(-r_n^2/8A^2) = 0
$$
holds where
$$
\del_{kn} = \inf\{ p_k(\mu) :  \mu\in \Re^k, \ \|\mu\| \le A+ r_n/2 \}, \ k=1,\ldots,m.
$$
Also assume that under the prior $\pi_2(\cdot|k)$ on $\bm{\sig}$, $Pr( \min(\bm{\sig}) < h_n|k)$ decays exponentially. 
Then under the Assumptions of Theorem~\ref{t1}, for any $\ep>0$, $k\ge 1$,
$$ E_{f_t}\big\{Pr\big( P(B_{r_n,k}^c) \ge \ep \big| k, \mv{X}_n \big)\big\} \lra 0.$$
If {\bf B1} is strengthed to 
$$
{\bf B1'}: \sum_{n=1}^\infty n\del_{kn}^{-1}h_n^{-k}\exp(-r_n^2/8A^2) < \infty,
$$
and the sequence $r_n$ satisfies  $\sum_{n=1}^\infty r_n^{-2(1+\alp)m} < \infty$ with $\alp$ as in Assumption {\bf A4}, then the conclusion can be strengthed to 
$$ \sum_{n=1}^\infty E_{f_t}\big\{Pr\big( P(B_{r_n,k}^c) \ge \ep \big| k, \mv{X}_n \big)\big\} < \infty.$$
\end{lemma}

With these three Lemmas we are now able to state and proof the theorem that ensures strong posterior consistency is attained.  

\begin{thm}\label{t4}
Consider a prior and sequences $h_n$ and $r_n$ for which the Assumptions of Lemma~\ref{l2} are satisfied. Further suppose that $n^{-1}(r_n/h_n)^m \lra 0$. 
Also assume that the sequence $r_n$  and  the prior $\pi_1(\cdot|k)$ on $(U,\th)$ satisfy the condition $Pr(\|\th\|>r_n|k)$ decays exponentially for $k \le m-1$.
Assume that the true density satisfies the conditions of Theorem~\ref{t1}.
Then the posterior probability of any $L_1$ neighborhood of $f_t$ converges to 1 in probability or almost surely depending on Assumption {\bf B1} or ${\bf B1'}$. 
\end{thm}

\begin{proof}
Theorem~\ref{t1} implies that the KL condition is satisfied. Consider the partition  $\mc{D}(\Re^m) = \mc{D}_n^{\ep} \cup \mc{D}_n^{\ep c}$. Then $N(\ep,\mc{D}_n^{\ep})/n \lra 0$. Write
$$
Pr(\mc{D}_n^{\ep c}| \mv{X}_n) = Pr \big( \{ f(.;\Th): \Th \in H_{n}^{\ep c} \} \big| \mv{X}_n \big),
$$
where
$$
H_n^{\ep c} = \{ \Th: \min(\bm{\sig}) < h_n \} \cup \{ \Th: \|\th\| > r_n \} \cup \{ \Th: P(B_{r_nk}^c)>\ep \}.
$$
The posterior probability of the first two sets above converge to 0 a.s. because the prior probability decays exponentially and the prior satisfies the KL condition.
Note that
$$
Pr\big(\{ \Th: P(B_{r_nk}^c)>\ep \} \big| \mv{X}_n \big) \le \sum_{j=1}^m Pr\big(\{ \Th: P(B_{r_nk}^c)>\ep \} \big| \mv{X}_n, k=j \big)
$$
and Lemma~\ref{l2} implies that this probability converges to 0 in probability/a.s. based on Assumption {\bf B1}/${\bf B1'}$.
Using Lemma~\ref{l3}, the result follows.
\end{proof}

Now we give an example of a prior that satisfies the conditions of Theorem~\ref{t4}. Any discrete distribution on $\{0,\ldots,m\}$ having $m$ in its support can be used as the prior $p$ for $k$.  Given $k$ ($k\ge 1$), we draw $U_0$ from a density on $V_{k,m}$. Given $k$ and $U_0$, under $\pi_1$, $\th$ is drawn from a density on the vector-space $\mc{N}(U_0)$ if $k<m$.  If $k=m$, then $\th=0$. When $k<m$, we set $\th = r\til\th$ with $r$ and $\til\th$ drawn independently from $\Re^+$ and the set $\{ \til\th \in \Re^m: \|\til\th\|=1, \til\th'U_0=0\}$ respectively. The scalar $r^a$ is drawn from a Gamma density for appropriate $a>0$. As a special case, a truncated normal density can be used for $\th$ when $\til\th$ is drawn uniformly, $a=2$ and $r^2 \sim Gam(1,\sig_0)$, $\sig_0>0$. Then $\th$ has the density
$$
\sig_0^{-(m-k)}\exp\frac{-1}{2\sig_0^2}\|\th\|^2 I(\th'U_0 = 0)
$$
with respect to the volume form of $\mc{N}(U_0)$.
Given $k$, $\bm{\sig}$ follows $\pi_2$ supported on $[0,A]^{k+1}$. Under $\pi_2$, the coordinates of $\bm{\sig}$ may be drawn independently with say, $\sig_j^{-2}$ following a Gamma density truncated to $[0,A]$.
If reasonable, assuming $\sig_1 = \ldots = \sig_k = \sig$ with $\sig^{-2}$ following a Gamma density will simplify computations. That said,  a Gamma distribution only satisfies the conditions of Theorem~\ref{t1} when $m\ge 2$.  To satisfy the conditions of Theorem \ref{t4} a \it{truncated transformed} Gamma density may be used. That is, for appropriate $b>0$, we draw  $\sig^{-b}$ from a Gamma density truncated to $[0,A]$.
Given $k$, $k\ge 1$, $P$ follows a $DP(w_kP_k)$ prior. To get conjugacy, we may select $P_k$ to be a Gaussian distribution on $\Re^k$ with covariance $\tau^2I_k$. 
With such a prior the conditions of Theorem~\ref{t4} are satisfied if we choose $a,b,\tau$ and $A$ such that $\tau^2 > 4A^2$, $a < 2(1+\alp)m$ and $a^{-1}+b^{-1} < m^{-1}$. This result is available from Corollary~\ref{t5} the proof of which is provided in the Appendix.

\begin{corr}\label{t5}
Assume that $f_t$ satisfies Assumptions {\bf A1-A4}. Let $\Pi_f$ be a prior on the density space as in Theorem~\ref{t4}. Pick positive constants $a,b,\{\tau_k\}_{k=1}^m$ and $A$ and set the prior as follows.
Choose $\pi_1(.|k)$ such that for $k \le m-1$, $\|\th\|^a$ follows a Gamma density. Pick $\pi_2(.|k)$ such that $\sig, \sig_1, \ldots, \sig_k$ are independently and identicaly distributed with 
$\sig^{-b}$ following a Gamma density truncated to $[0,A]$. Alternatively let $\sig=\sig_1=\ldots=\sig_k$ with $\sig$ distributed as above.
For the $DP(w_kP_k)$ prior on $P$, $k\ge 1$, choose $P_k$ to be a normal density on $\Re^k$ with covariance $\tau_k^2I_k$. Then almost sure strong posterior consistency results if the constants 
satisfy
$\tau_k^2 > 4A^2$, $a < 2(1+\alp)m$ and $1/a + 1/b <  1/m$.
\end{corr}

A multivariate gamma prior on $\bm{\sig}$ satisfies the requirements for weak but not strong posterior consistency (unless $m=1$). However that does not prove that it is not eligible because
Corollary~\ref{t5} provides only sufficient conditions. Truncating the support of $\bm{\sig}$ is not undesirable because for more precise fit we are interested in low within cluster covariance which will result in sufficient number of clusters.
However the transformation
power $b$ increases with $m$ resulting in lower probability near zero which is undesirable when sample sizes are not high. 

In \cite{abhishek2}, a gamma prior is proved to to be eligible for a Gaussian mixture model (that is, $k=m$) as long as the hyperparameters are allowed to depend on sample size in a suitable way.
However there it is assumed that $f_t$ has a compact support. We expect the result to hold true in this context too.

\section{Identifiability of Parameters}
In many applications, the goal may not be density estimation but estimating the low dimensional set $S$ and its dimension. To do so $S$ must be identifiable.   That is, there must be a unique $S$ corresponding to the model \eqref{e7}.
Denoting by $P_f$, the distribution corresponding to $f$, it follows that
\begin{equation} \label{e8}
P_f = N_m(0,\Sig)* (P\circ\phi^{-1}),
\end{equation}
with * denoting convolution.  Now let $\Phi_P(t)$ be the characteristic function of a distribution $P$, then \eqref{e8} implies that the characteristic function of $f$ (or $P_f$) is
\begin{equation} \label{e12}
\Phi_f(t) = \exp(-1/2 t'\Sig t) \Phi_{P\circ \phi^{-1}}(t), \ t\in \Re^m.
\end{equation}
Once we let $P$ to be discrete, \eqref{e12} suggests that $\Sig$ and $P\circ\phi^{-1}$ can be uniquely determined from $f$.
Now $\phi: \Re^k \lra \Re^m$, $\phi(\Re^k) = S$ and $P\circ \phi^{-1}$ is the distribution of $\phi(Y)$ with $Y \sim P$. It is a distribution on $\Re^m$ supported on the $k$ dimensional affine plane $S$.
To identify $S$ and $k$, we further assume that the \it{affine support} asupp$(P)$ of $P$ is $\Re^k$.
We define asupp$(P)$ as the intersection of all affine subspaces of $\Re^k$ having probability 1. It is an affine subspace containing supp$(P)$ (but may be larger).
In other words, we use a prior for which $P$ is discrete and asupp$(P) = \Re^k$ w.p. 1. The Dirichlet process prior on $P$ given $k$ with a full support base is an appropriate choice.
Then, from the nature of $\phi$, asupp$(P\circ\phi^{-1})$ is an affine subspace of $\Re^m$ of dimension equal to that of asupp$(P)$.
Since asupp($P\circ\phi^{-1}$) is identifiable, this implies that $k$ is also identifiable as its dimension.
Since $S$ contains asupp($P\circ\phi^{-1}$) and has dimension equal to that of asupp($P\circ\phi^{-1}$), hence $S = \mr{asupp}(P\circ\phi^{-1})$.
Hence we have shown that the (sub) parameters $(\Sig,k,S, P\circ\phi^{-1})$ are identifiable once we set a full support discrete prior on $P$ given $k$. Then $U_0 U_0'$ and $\th$ are identifiable as the projection matrix and origin of $S$. 
However $P$ and the coordinate choice $\phi$ (hence $U_0$) are still non-identifiable.  However, if we consider the structure $\Sig = U_0\Sig_0U_0' + \sig^2 (I_m - U_0 U_0')$ with a diagonal $\Sig_0$ and impose some ordering on the diagonal entries of $\Sig_0$, then the columns of $U_0$ become identifiable up to a change of signs as the eigen-rays.

\subsection{Point estimation for subspace $S$}
To obtain a Bayes estimate for the subspace $S$, one may choose an appropriate loss function and minimize the Bayes risk defined as the expectation of the loss over the posterior distribution.
Any subspace is characterized by its projection matrix and origin.  That is, the pair $(R,\th)$ where $R\in M(m)$ and $\th\in \Re^m$ satisfy $R=R'=R^2$ and $R\th=0$. We use $\mc{S}_m$ to denote the space of all such pairs. One particular loss function on $\mc{S}_m$ is
$$ 
L_1((R_1,\th_1), (R_2, \th_2)) = \|R_1 - R_2 \|^2 + \|\th_1 - \th_2\|^2, \ (R_i,\th_i) \in \mc{S}_m.
$$
For a matrix $A = ((a_{ij}))$, its norm-squared is defined as $\|A\|^2 = \sum_{ij} a_{ij}^2 = \mr{Tr}(AA')$.
We find the average of $L_1$ over repeated draws of $(R_2, \th_2)$ from their posterior and choose the value of $(R_1,\th_1)$ for which the average is minimized (if a unique minimizer exists). Then the subspace $S$ is estimated as
$\{ R_1x + \th_1: x \in \Re^m \}$. It has dimension equal to the rank of $R_1$.

If the goal is to estimate the directions of the subspace, we may instead use the loss function
$$
L_2((U_1, w_1), (U_2, w_2)) = \| U_1 - U_2 \|^2 + (w_1 - w_2)^2, \ (U_i,w_i) \in \mc{S}_{m2}.
$$ 
Here the $m\times m$ matrix $U_i$ has the first few columns as the directions of the corresponding subspace $S_i$, the next column gives the direction of the subspace origin $\th_i$ and the rest are set to
the zero vector while $w_i = \|\th_i\|$. Therefore 
$$
\mc{S}_{m2} = \left\{ (U, w)\in M(m)\times\Re^+: \ U'U = \left( \begin{array}{cc}
                                                                                    I & 0 \\
                                                                                    0 & 0
                                                                                    \end{array} \right) \right\}.
$$
We find the minimizer (if unique) $(U_1, w_1)$ of the expected value of $L_2$ under the posterior distribution of $(U_2,w_2)$ and set the estimated subspace dimension $k$ as the rank of $U_1$ minus 1, the principal directions consisting of the first $k$ columns of $U_1$ and the origin as $w_1$ times the last column.
Since the $k$ orthonormal directions of the subspace are only identifiable as rays, one may even look at the loss
\begin{align*}
L_3((U,\th_1), (V,\th_2)) = \sum_{j=1}^m \| U_j U_j' - V_j V_j'\|^2 + \|\th_1 - \th_2\|^2,
\end{align*}
where
\begin{align*}
&(U,\th_1),(V,\th_2) \in \mc{S}_{m3} = \left\{ (U, \th)\in M(m)\times\Re^m: \ U'U = \left( \begin{array}{cc}
                                                                                    I & 0 \\
                                                                                    0 & 0
                                                                                    \end{array} \right), \ U'\th = 0 \right\}.
\end{align*}
Theorems~\ref{t2} and \ref{t3} (proofs of which can be found in the appendix) derive the expression for minimimizer of the risk function corresponding to $L_1$ and $L_2$ and present conditions their uniqueness.
Hereby we denote by $P_n$ the posterior distribution of the parameters given the sample. It is assumed to have finite second order moments. For a matrix $A$, by $A_{(k)}$ we shall denote the submatrix of $A$ consisting of its first $k$ columns.

\begin{thm}\label{t2}
Let $f_1(R,\th) = \int_{(R_2,\th_2)}L_1((R,\th), (R_2, \th_2)) dP_n(R_2,\th_2)$, $(R,\th) \in \mc{S}$. This function is minimized by
$R = \sum_{j=1}^k U_j U_j'$ and $\th = (I - R)\bar\th_2$ where
$\bar R_2 = \int_{M(m)} R_2 dP_n(R_2)$ and $\bar\th_2 = \int_{\Re^m} \th_2 dP_n(\th_2)$ are the posterior means of $R_2$ and $\th_2$ respectively,
$2\bar R_2 - \bar\th_2\bar\th_2' = \sum_{j=1}^m \l_j U_j U_j'$, $\l_1 \ge \ldots \ge \l_m$ is a s.v.d. of $2\bar R_2 - \bar\th_2\bar\th_2'$, and $k$ minimizes
$k - \sum_{j=1}^k \l_j$ on $\{0,\ldots,m\}$. The minimizer is unique if and only if there is a unique $k$ minimizing  $k - \sum_{j=1}^k \l_j$ and $\l_k > \l_{k+1}$ for that $k$.
\end{thm}

\begin{thm}\label{t3}
Let $f_2(U,w) = \int_{(U_2,w_2)}L_2((U,w), (U_2, w_2))dP_n(U_2, w_2)$, $(U,w) \in \mc{S}_{m2}$. Let $\bar w$ and $\bar U$ denote the posterior means of $w_2$ and $U_2$  respectively.
Then $f_2$ is minimized by
$w = \bar w$ and any $U = [U_1, 0]$, where $U_1 \in V_{k+1,m}$ satisfys
$\bar U_{(k+1)} = U_1 (\bar U_{(k+1)}'\bar U_{(k+1)})^{1/2}$, and $k$ minimizes
$g(k) = k - 2\mr{Tr}(\bar U_{(k+1)}'\bar U_{(k+1)})^{1/2}$ over $\{0,\ldots,m-1\}$.
The minimizer is unique if and only if there is a unique $k$ minimizing  $g$ and $\bar U_{(k+1)}$ has full rank for that $k$.
\end{thm}

\section{Posterior Computation}\label{s4}
We now present an algorithm to sample from the joint posterior distribution of $\Th = (k,U_0,\th, \Sigma_0, \sig, P)$ and as a result the density of $X$, given iid realizations $X_1, \ldots, X_n$. Since exact sampling is not possible, we resort to MCMC draws from the posterior.   We first present an algorithm with $k$ being treated as a fixed known quantity.  We then generalize the algorithm to allow unknown $k$.  In both cases,  a straight forward Gibbs sampler can be used.  


\subsection{MCMC algorithm for the fixed $k$}\label{s5.1}
We use  a Dirichlet process (DP) prior for $P$ (i.e., $P \sim DP(w_0P_0))$.   For simplicity and to preserve conjugacy we set $P_0 = N_k(m_{\mu}, S_{\mu})$ with $w_0 = 1$.  We employ the stick breaking representation of the Dirichlet process (Sethuraman \cite{sethuraman}) so that $P = \sum_{j=1}^\infty w_j \del_{\mu_j}$ where $\mu_j$ is drawn $iid$ from $P_0$ and $w_j = v_j\prod_{\ell < j}(1-v_{\ell})$ with $v_j \sim Beta(1, w_0)$.  After introducing cluster labels $S_1,\ldots,S_n$, the likelihood becomes
\begin{align}\label{e18}
f(\bm{x};  U_0, \theta, \Sigma_0, \sigma, P, \mu, S) = & \prod_{i=1}^nw_{S_i}N_m(x_i; U_0\mu_{S_i} + \theta, \Sigma) \\
							 = & \prod_{i=1}^nw_{S_i}N_k(U'_0x_i; \mu_{S_i}, \Sigma_0)N_{m-k}(V'x_i; V'\theta, \sigma^2I_{m-k}) 
\end{align}
where once again $\Sig = U_0\Sigma_0U_0' + \sig^2(I_m - U_0U_0')$.  After prior distributions for $(U_0, \theta, \Sigma_0, \sigma, \mu)$ are appropriately selected  (details of which are given concurrently within the description of the algorithm) it is now possible to describe an algorithm that can be used to construct an MCMC chain that provides draws from the joint posterior distribution of interest by cycling through the following steps.

\begin{enumerate}
\item[{\bf Step 1}.]
Let $\pi(U_0)$ denote a prior distribution for $U_0 \in V_{k,m}$. Using straightforward matrix algebra it can be shown that the full conditional of $U_0$ is 
\begin{align}\label{e101}
[U_0 | -] & \propto \exp\{tr\big[1/2(\sig^{-2}I_k - \Sigma_0^{-1})U_0' (\sum_{i=1}^n x_i x_i')U_0 + \Sigma_0^{-1}(\sum_{i=1}^n \mu_{S_i} x_i')U_0\big]  \}\pi(U_0) \nonumber \\
	& \propto \mr{etr}\{F_1'  U_0 + F_2  U_0' F_3  U_0 \}\pi(U_0), 
\end{align}
\end{enumerate}
where $F_1  = (\sum_{i=1}^n x_i \mu_{S_i}')\Sigma_0^{-1}$,  $F_2  = \frac{1}{2}(\sig^{-2}I_k - \Sigma_0^{-1})$, and $F_3  = \sum_{i=1}^n (x_i x_i')$.  In \eqref{e101} ${\rm etr}(A)$ denotes $\exp(tr(A))$.  Thus,  if one selects a matrix Bingham-von Mises-Fisher prior distribution for $U_0$ (the Uniform distribution on the Steifel manifold being a special case), then the full conditional of $U_0$ is a matrix Bingham-von Mises-Fisher distribution on the space $U'_0\theta = 0$.   Strategies for sampling from matrix Bingham-von Mises-Fisher are developed in Hoff \cite{hoff2}.  A straightforward extension of their work can be implemented to sample from a matrix Bingham-von Mises-Fisher that has $U'_0\theta = 0$ as a constraint.

\begin{enumerate}
\item[{\bf Step 2}.]
As discussed in Section 3.2 a good prior choice for $\theta$ is a truncated normal $\theta \sim N_m(m_{\theta}, S_{\theta})I[U'_0\theta = 0]$.  The full conditional under this prior is the following truncated multivariate normal
\begin{align}\label{5.4}
[\th | -]\sim N_m (m_{\theta}^*, S_{\theta}^*)I[U_0'\th = 0],
\end{align}
where  $S_{\theta}^* = (n\Sigma^{-1} + S^{-1}_{\theta})^{-1}$ and $m_{\theta}^* = S_{\theta}^*(\Sigma^{-1} \sum_{i=1}^n x_i + S^{-1}_{\theta}m_{\theta}).$   
\end{enumerate}
Notice that if $W$ is an orthonormal basis of $\mc{N}(U'_0)$, then there exists a $\tilde{\theta} \in \Re^{m-k}$ such that $\theta = W\tilde{\theta}$ and $\tilde{\theta} \sim N_{m-k}(W'm_{\theta}^*, W'S_{\theta}^*W)$.   
This fact can be exploited to sample from \eqref{5.4}.

\begin{enumerate}
\item[{\bf Step 3}.]
Update $S_i$ for $i = 1, 2, \ldots, n$ by sampling from the multinomial conditional posterior distribution
\[
Pr(S_i = j|-) \propto w_j \exp\{-1/2 (U'_0 x_i - \mu_j)' \Sig_0^{-1} (U' _0x_i - \mu_j) \}, \ j=1,\ldots,\infty.
\]
\end{enumerate}
To make the total number of states finite the block Gibbs sampler of Ishwaran and James \cite{ish} may be implemented.  Alternatively, the  slice sampling ideas described in Yau, Papaspiliopoulos, Roberts, and Homes \cite{yua2011},  Walker  \cite{walker2007}, or Kalli, Griffin, and Walker  \cite{griffin2011} could be used.  The remainder of the algorithm is described from the perspective of using a block Gibbs sampler which requires truncating the number of atoms to $N$.

\begin{enumerate}
\item[{\bf Step 4}.]
Update the DP atom weights by setting $w_j = v_j \prod_{l=1}^{j-1} (1- v_l)$, $j=1,\ldots, N$ after drawing
$$[v_l | -] \sim Beta(1 + n_j, w_0 + \sum_i I(S_i > j))$$ 
with $n_j = \sum_i I(S_i = j)$ and setting $v_N = 1$.
\end{enumerate}

\begin{enumerate}
\item[{\bf Step 5}.]
Update the DP atoms $\{\mu_j: j =1, \ldots, N\}$ independently by sampling from 
\[
[\mu_j | - ] \sim N_k(m_{\mu}^*, S_{\mu}^*),
\]
where $S_{\mu}^* = (n_j\Sigma_0^{-1} + S_0^{-1})^{-1}$ and $m_{\mu}^* = S_{\mu}^*(U'_0 \Sigma_0^{-1} \displaystyle \sum_{i:S_i=j} x_{i} + S^{-1}_{\mu} m_{\mu})$. 
\end{enumerate}

\begin{enumerate}
\item[{\bf Step 6}.]
Using a $\sigma^{-2} \sim {\rm Ga}(a,b)$ prior,  $\sigma^{-2}$ can be updated using
\[
[\sigma^{-2}|-] \sim {\rm Ga}(\frac{1}{2}n(m-k) + a, b + \frac{1}{2}\sum_{i=1}^nx'_ix_i + \frac{n}{2}\theta'\theta - \frac{1}{2}\sum_{i=1}^nx'_iU_0U'_0x_i - \theta'\sum_{i=1}^nx_i)
\]
Under the simplifying assumption that $\Sigma_0 = \sigma^{2}I_k$ the full conditional of $\sigma^{-2}$ becomes
\[
[\sigma^{-2}|-] \sim {\rm Ga}(\frac{1}{2}nm + a, b + \frac{1}{2}\sum_{i=1}^n(x_i  - U_0\mu_{S_i} - \theta)'(x_i  - U_0\mu_{S_i} - \theta))
\]
\end{enumerate}

\begin{enumerate}
\item[{\bf Step 7}.]
Using a truncated Gamma distribution for $\sigma^{-2}_j$ (i.e., $\sigma^{-2}_j \sim Gam(a, b)I[\sigma^{-2}_j \in [0,A]]$) allows one to update $\sigma^{-2}_j$ using the following truncated Gamma distribution.
\[
[\sigma^{-2}_{j}|-] \sim {\rm GAM}(\frac{n}{2} + a, b + \frac{1}{2}\sum_{i=1}^n(U_0'x_i - \mu_{S_i})^2_{j})I[\sigma^{-2}_j \in [0,A]].
\]  
\end{enumerate}

Reasonable starting values can decrease the number of MCMC iterates discarded as burn in and therefore may be desirable.  For $U_0$,  the first $k$ eigen-vectors of the sample covariance matrix can be used.  For $\theta$ one may use $(I_m - U_sU'_s)\bar{x}$ where $U_s$ denotes the starting value for $U_0$.  The initial labels $(S_i)$ and coordinate cluster means ($\mu_j$) can be obtained by applying a k-means algorithm to $U'_s x_i$.


\subsection{MCMC algorithm for $k$ unknown}
In the case that $k$ is unknown,  a prior distribution needs to be assigned to $k$ and  $U_0 \in O(m)$.   In what follows, to denote  the $k$th coordinate and the 1st $k$ coordinates of  $\mu_j$ we use $\mu_{jk}$ and $\mu_{j(k)}$ respectively.   Similarly,  let $U_{0(k)}$ represent the first $k$ columns of $U_0$ while $U_{0(-k)}$ will represent the remaining $m-k$ columns. 

After introducing cluster labels, the full posterior is proportional to
$$
\pi(w,\mu,\sig,\Sig_0, U_0,\theta, k, S) \propto \prod_{i=1}^n w_{S_i} N_m \big( x_i; U_{0(k)}\mu_{S_i(k)} + \th,  \Sig \big).
$$
Here $\pi$ is a general expression for the prior.  The first $k$ columns of the $m\times m$ matrix $U_0$ explain the subspace directions and the first $k$ coordinates of $\mu_j$ the cluster locations.  

Allowing $k$ to be unknown requires altering steps 1 and 5 of the MCMC algorithm described in the previous section and adding an additional step.   We first describe the additional step and then the adjustments to steps 1 and 5.  Continuing from step 7 from the previous section we add

\begin{enumerate}
\item[{\bf Step 8}.]
Update $k$ by drawing a value for $k$ from the following complete conditional  
\begin{align} \label{s8}
Pr(k = \ell |-) & \propto p(\ell)\prod_{i=1}^n N_{m}(x_i; U_{0(\ell)}\mu_{S_i(\ell)} + \th,  \Sig) \ \mbox{for $\ell = 1, \ldots, m-1$}.
\end{align}
\end{enumerate}


When the data dimension $m$ is very high, computing all $m-1$ probabilities can become computationally expensive.  An approach to reduce the number of states would be to introduce a slice sampling variable $u$ drawn from $Unif(0,1)$.    In this setting we replace $p(k)$ in \eqref{s8} by $I(u< p(k))$.  This means that $k$ will be drawn from the set $\{k\colon p(k)>u\}$ and $u \sim Unif(0,p(k))$.  Updating the upper bound for the subspace dimension  ($K$) can be done by drawing $u \sim Unif(0,p(k))$ and setting $K = \max\{k \le m: p(k) > u)\}$.   
\begin{enumerate}
\item[{\bf Step 1b.}]
Use the complete conditional derived in step 1 from Section 6.1 to update $U_{0(k)}$, then draw $U_{0(-k)} = [U_{0k+1}, \ldots, U_{0K}]$ from  $\pi(U_{0(-K)} | U_{0k})$ such that $U'_{0(-k)}\theta = 0$. 

\end{enumerate}


When a uniform prior is being considered, step1b requires one to sample uniformly from $V_{K-k,m}$ perpendicular to the column space of $[U_{0k}, \theta] \equiv U_\theta$. As discussed in Chikuse\cite{chikuse}, $U^*$ is a uniform sample from $V_{K-k,m}$ if  $U^* = T(T'T)^{-1/2}$ for $T$  a $m\times(K-k)$ matrix of independent standard normal random variables.  To ensure that $U^* \in  \mc{N}(U_\theta')$ first project $T$ into $\mc{N}(U_\theta')$ by  setting $T^* = (I - U_\theta U_\theta')T$.  Then $U^* = T^*(T^{*'}T^*)^{-1/2}$ is a uniform draw from $V_{K-k,m}$ perpendicular to column space of $U_\theta$.  If $\pi(U_0)$ is not a uniform distribution on $O(m)$ see Hoff \cite{hoff2} for sampling strategies.
 
\begin{enumerate}
\item[{\bf Step 5b.}]
Use the full conditional found in step 5 from Section 6.1 to update $\mu_{j(k)}$.  Then draw $\mu_{jk+1}, \ldots, \mu_{jK}$ from their respective prior distributions. 
\end{enumerate}

With $k$ unknown, the MCMC chain tends to get stuck on certain values of $k$ for many iterations.  The stickiness occurs because the probabilities in step 8 are computed for all $\ell = 1, \ldots, K$ using a $U_0$ that was updated for a particular value of $k$.   To make the chain less sticky, we employ adaptive MCMC methods as outlined in Roberts and Rosenthal \cite{roberts}.  We applied the adaptation to step 8 and step 5 of the algorithm.  Specifically, we raised each of the un-normalized probabilities in \eqref{s8} to the  $1 - \exp(-0.0001t)$ power (where $t = 1, \ldots, M$ denotes the $t^{\rm th}$ MCMC iterate) and replace $S_{\mu}^*$ found in step 5 of Section~\ref{s5.1} with $(1+100\exp(-0.001t))S_{\mu}^*$.  In this way,  the space of cluster locations is initially more thoroughly explored.  Notice that the adaptation vanishes at an exponential rate, which guarantees that the proper regularity conditions hold.


\section{Simulation Study}


To assess the proposed methodology's density estimation ability we conducted a small simulation in which a density is estimated using observations in $\Re^m$ originating from the following finite mixture 
\begin{align}\label{e61}
\bm{x} \sim \sum_{h=1}^{c+1} \pi_hN_{m}(\bm{\eta}_h, \sigma^2I). 
\end{align}
Here $\eta_h$ is a vector of zeros save for the $h$th entry which is 1.    We considered the following three factor's influence on the density estimate.
\begin{enumerate}
\item Bandwidth (setting $\sigma^2=0.01$, $\sigma^2=0.05$, and $\sigma^2=0.1$)
\item Sample size (setting $n=50$, $n=100$, $n=200$)
\item Dimension of the affine subspace (considering $k=2$ and $k=5$).
\end{enumerate}

To show that \eqref{e61} falls into the current class of models, consider the case of $k=2$ and $m=100$.  For this case we have the $100$-dimensional vector $\theta = (1/3, 1/3, 1/3, 0, \dots, 0)'$.  
Further one possible representation of the $100\times2$ dimensional $U_0$ is
\begin{align}\label{e62}
U_0 = 
\left(
\begin{array}{ccccc}
  1/\sqrt{2}   & -1/\sqrt{2} & 0 & \ldots & 0 \\
  1/\sqrt{6}   &  1/\sqrt{6} & -2/\sqrt{6} & \ldots & 0
\end{array}
\right)'.
\end{align}

As competitors, we considered a finite mixture with $f(x) = \sum_{h=1}^c\pi_hN_m(\mu_h, \sigma^2\bm{I}_m)$   and  an infinite mixture $f(x) = \sum_{h=1}^{\infty}\pi_hN_m(\mu_h, \sigma^2\bm{I}_m)$.  The number of components employed in the finite mixture were 3 and 6 for the two respective affine subspace dimensions considered.  For each synthetic data set created, 100 observations  were generated to assess out of sample density estimation.  To compare the density estimates between the procedures employed,  we used the following Kullback-Leibler type distance  
\begin{align} \label{e20}
\frac{1}{D} \sum_{d=1}^D\frac{1}{T} \sum_{t =1}^T \left ( \sum_{\ell = 1}^{100} \log{f_0(\bm{x}^*_{\ell d})} - \sum_{\ell = 1}^{100}\log{\hat{f}_t(\bm{x}^*_{\ell d} )}\right). 
\end{align}
Here $f_0$ denotes the true density function, $d$ is an index for the $D=25$ datasets that were generated, and $\bm{x}^*_{\ell d}$ is the $\ell$th out of sample observation generated from the $d$th data set and $\hat{f}_t$ is the estimated density.

For each of the 25 generated data sets, a density estimate was obtained using the proposed method with $k$ unknown and for $k=1$, $k=2$, and $k=5$.   We entertained a discrete uniform and stick-breaking type prior for $k$ with no appreciable difference in parameter estimation.  We set $\sigma_1 = \ldots, \sigma_k = \sigma$.   For each scenario 1000 MCMC iterates were used to approximate the density.  A burn-in of 1000 was used when $k$ was fixed.   When $k$ was considered an unknown a burn-in of 10,000 was used with a thin of 100.   Convergence was monitored using trace plots of the collected MCMC iterates.  

The value of equation \eqref{e20} for each scenario considered averaged across the 25 datasets can be found in Table \ref{SimStudyResults}.    Under the column ``Unknown $k$" can be found the results when $k$ was treated as an unknown.  The results from the method when $k$ is fixed at a specified value can be found under one of the three ``$k=$" columns.  Results from the finite mixture and infinite mixture are under the columns ``Fin Mix" and ``Inf Mix".

\begin{table}[htdp]
\caption{Results of the Kullback-Liebler type distance comparing estimated densities from each of the procedures considered in the simulation study to the density used to generate data.}
\begin{center}
\begin{tabular}{cc ccccccc}
True $k$ & $\sigma^2$ & $n$ &  Unknown $k$  & $k=1$ &  $k=2$ & $k=5$ & Fin Mix & Inf Mix\\ \midrule 
\multirow{9}{*}{2} & \multirow{3}{*}{0.01} & 50 	& 582.98 & 1557.39 & 392.84 & 412.77 & 2580.81 & 2612.92 \\ 
&							       & 100 & 274.76 & 1494.65 & 205.49 & 214.32 & 1539.74 & 1619.44 \\ 
&							       & 200 & 139.21 & 1474.90 & 106.06 & 111.85 & 165.92 & 1429.98 \\ 		\cmidrule(lr){2-9}
& \multirow{3}{*}{0.05}                                &  50 	& 590.24 & 421.93 & 314.44 & 394.53 & 710.46 & 714.26 \\ 
&							       &100 	& 271.79 & 371.65 & 172.61 & 192.39 & 465.87 & 499.58 \\ 
&							       & 200 & 128.30 & 315.85 & 96.37 & 105.34 & 153.54 & 160.66 \\ 		\cmidrule(lr){2-9}
& \multirow{3}{*}{0.1}                                  &  50 	& 589.01 & 232.33 & 250.50 & 365.38 & 426.69 & 426.29 \\ 
&							       &100 	& 280.99 & 189.05 & 154.91 & 201.62 & 320.02 & 324.92 \\ 
&							       & 200 & 134.07 & 162.34 & 87.55 & 104.65 & 160.54 & 176.29 \\ 		\midrule
\multirow{9}{*}{5} & \multirow{3}{*}{0.01} & 50 	& 2292.44 & 2645.34 & 2268.70 & 1015.80 & 3003.87 & 3029.25 \\  
&							       & 100 & 2075.99 & 2564.26 & 2164.32 & 500.65 & 2341.99 & 2838.46 \\ 
&							       & 200 & 2138.87 & 2503.26 & 2065.54 & 256.78 & 1646.43 & 2046.68 \\ 	\cmidrule(lr){2-9}
& \multirow{3}{*}{0.05}                                & 50  	& 872.18 & 646.12 & 654.20 & 714.96 & 798.29 & 801.22 \\ 
&							       & 100 & 604.07 & 604.73 & 556.36 & 421.40 & 676.65 & 690.04 \\ 
&							       & 200 & 506.53 & 550.92 & 489.39 & 231.47 & 460.85 & 512.93 \\ 		\cmidrule(lr){2-9}
& \multirow{3}{*}{0.1}                                  & 50  	& 773.15 & 315.85 & 357.02 & 484.87 & 447.79 & 456.62 \\ 
&							       & 100 & 431.56 & 294.42 & 309.34 & 358.66 & 351.17 & 353.89 \\ 
&							       & 200 & 283.02 & 246.20 & 237.94 & 206.01 & 286.96 & 288.10 \\ 
\bottomrule
\end{tabular}
\end{center}
\label{SimStudyResults}
\end{table}%

Generally speaking, the procedure outlined in Section 3 does a much better job at recovering the true density relative to the mixtures.  This is the case even if $k$ is fixed at the wrong value.  That said,  as expected,  fixing $k$ at the true value provides the best results.  The only instances in which the finite mixture estimated the density more accurately than our density estimator is when the dimension of the affine subspace is set to 5 and the sample size is small.  However, even in small samples, if $k$ is fixed at the correct value, then the density is recovered more accurately using our procedure compared to mixtures.    Also, it appears as $\sigma^2$ increases, then cluster separation diminishes and estimating $k$ is more difficult.   Hence the varying $k$ procedure does not perform as well in estimating the density (which is to be expected) but still out performs the mixtures.  In addition, as expected larger sample sizes are conducive to better density estimation as the Kullback-Leibler type distance generally gets smaller as $n$ increases.  






\section{Nonparametric Classification with Feature Coordinate Selection} \label{s5}
We consider a categorical $Y$ that takes on values from the set $\{1,\ldots,c\}$.  The goal of classification is to identify the class to which $Y$ belongs using $m$ characteristics of $Y$.  These characteristics are typically denoted by $X \in \Re^m$.  Because the association between $X$ and $Y$ may not be causal, our approach is to model $X$ and $Y$ jointly and from the joint derive the conditional.    Letting $M_c(y; \bm{\nu}) = \prod_{\ell = 1}^c \nu_{\ell}^{I[y=\ell]}$, we consider the following joint model  
\begin{equation}\label{e13}
(X,Y) \sim f(x,y) = \int_{\Re^k\times S_{c}}  N_m( x;\phi(\mu), \Sig )  M_c(y; \bm{\nu}) P(d\mu \, d\bm{\nu}),
\end{equation}
with $S_{c} = \{ \bm{\nu}  \in [0,1]^c \colon \sum \nu_{\ell} =1 \}$ denoting the $c-1$ dimensional simplex.  Note that \eqref{e13} is a generalization of \eqref{e7} and \eqref{e90} along the lines of the joint model proposed  in Bhattacharya  and Dunson \cite{abhishek1},   though they focus on kernels for predictors on models that accommodate non-Euclidean manifolds and there is no dimensionality reduction.

When $m$ is large it is often the case that most of the information present in the data is used to model the marginal of $X$ while the association between $X$ and $Y$ is disregarded.  In order to avoid this,  we instead pick a few coordinates of $X$, say $k$ many, and model the joint density of the $k$ coordinates of $X$ and $Y$. The remaining coordinates of $X$ are modeled independently as equal variance Gaussians, though in preliminary simulation studies, we find that our performance in estimating the subspace and predicting $Y$ is robust to the true joint distribution of the `non-signal' predictors that are not predictive of $Y$.  By setting a prior on the coordinate selection method, we can pick out those few `important' coordinates which completely explain the conditional distribution of $Y$, very flexibly. 
Without loss of generality an isotropic transformation on $X$ can be used which would provide some benefit with regards to coordinate inversion.  That is, we can locate a $k\le m$ and $U_0 \in V_{k,m}$ such that
\begin{align}
(U'_0X, Y) \sim f_1(x_1,y) = \int_{\Re^k\times S_{c}}  N_{k}(x_1;\mu,\Sig_0)  M_c(y; \bm{\nu}) P(d\mu \, d\bm{\nu}), \ x_1 \in \Re^k,
\end{align}
along with a $\th \in \Re^m$ and $V \in V_{m-k,m}$ satisfying $V'U_0 = 0$ and $\th'U_0 = 0$, such that
\begin{equation}\label{e14}
V'X \sim N_{m-k}(V'\theta, \sig^2 I_{m-k})
\end{equation}
independently of $(U'_0X, Y)$.
With such a structure, the joint distribution of $(X,Y)$ becomes \eqref{e13}
where
\begin{eqnarray*}
\phi:\Re^k \ra \Re^m, \ \phi(y) = U_0y + \th, \ U_0 \in V_{k,m}, \ \th \in \Re^m, \ U'_0\th = 0,\\
\Sig = U_0(\Sig_0 - \sig^2I_k) U'_0 + \sig^2 I_m, \ \Sig_0 \in M^+(k), \sig^2 \in \Re^+.
\end{eqnarray*}
The conditional density of $Y=y$ given $X=x$ can be expressed as
\begin{equation}\label{e16}
p(y|x;\Th) = \frac{\int_{\Re^k\times S_{c}}  N_k(U'_0x;\mu,\Sig_0) M_c(y; \bm{\nu}) P(d\mu \, d\bm{\nu})} {\int_{\Re^k\times S_{c}} N_k(U'_0x;\mu,\Sig_0)P(d\mu \, d\bm{\nu})}
\end{equation}
with parameters $\Th = (k,U_0,\Sig_0, P, \theta, \sigma^2)$. A draw from the posterior of $\Th$ given model \eqref{e13} will give us a draw from the posterior of the conditional. 
When $P$ is discrete (which is a standard choice), the conditional distribution of $Y$ given $X$ and $\Th$ can be thought of as a weighted $c$ dimensional multinomial probability vector with the weights depending on $X$ only through the selected $k$-dimensional coordinates $U'_0X$. 
For example, if $P = \sum_{j=1}^\infty w_j \del_{(\mu_j,\bm{\nu}_j)}$, then
\begin{align} \label{e50}
p(y|x;\Th) = \sum_{j=1}^\infty \til w_j(U'_0x)  M_c(y; \bm{\nu}_j)
\end{align}
where
$\til w_j(x) = \frac{w_j N_k(x;\mu_j,\Sig_0)}{\sum_{i=1}^\infty w_i N_k(x;\mu_i,\Sig_0)}$ and $x \in \Re^k$ for $j=1,\ldots,\infty$.   
We refer to \eqref{e50} as the principal subspace classifier (PSC).

The above is easily adapted to a regression setting by considering a low dimensional response $Y \in \Re^l$ and replacing the multinomial kernel used for $Y$ with a Gaussian kernel.  In this setting the joint model becomes
\begin{eqnarray}\label{e17}
(X,Y) \sim \int_{\Re^k\times\Re^l} N_m(x;\phi(\mu),\Sig_x)N_l(y;\psi,\Sig_y)P(d\mu \, d\psi),
\end{eqnarray}
which produces the following conditional model
\begin{equation}\label{e18}
p(y|x;\Th) = \frac{\int_{\Re^k\times\Re^l} N_k(U'_0x;\mu,\Sig_0) N_l(y;\psi,\Sig_y)P(d\mu \ d\psi)} {\int_{\Re^k\times\Re^l} N_k(U_0'x;\mu,\Sig_0)P(d\mu \, d\psi)}.
\end{equation}
For a discrete $P$ this conditional distribution becomes the following mixture whose weights depend on $X$ only through its $k$-dimensional coordinates $U'_0X$
\begin{align}
p(y|x;\Th) = \sum_{j=1}^\infty \til w_j(U'_0x) N_l(y;\psi_j,\Sig_y).
\end{align}
As the regression model is a straightforward modification of the classifier, we focus on the classification case for sake of brevity.

\subsection{MCMC algorithm}
Sampling from the posterior of $\Th = (k,U_0,\Sig_0, P, \theta, \sigma^2)$ requires adjusting step 3 of Section 6's algorithm and adding a step to update $\bm{\nu}$.   We continue to assume $P \sim DP(\alpha, P_0)$.  However,  in the present setting $P_0 = N(m, S)\otimes Dir(\bm{a}_{\nu})$.   Now the data likelihood, after introducing cluster labels $S_1,\ldots,S_n$, becomes $\prod_{i=1}^n w_{s_i} N_m(x_i; U\mu_{S_i}+\th, \Sig_0)M_c(y_i; \bm{\nu}_{S_i})$.  An MCMC chain that provides draws from the joint posterior of $\Theta$ can be obtained by adding the following two steps to the algorithm in Section 6. 

\begin{enumerate}
\item[{\bf Step 3}.] 
Update $S_i$ for $i = 1, 2, \ldots, n$ by sampling from the following conditional posterior distribution
\[
Pr(S_i = j|-) \propto w_j \exp \left\{-1/2 (\mu_j'\Sig_0^{-1}\mu_j -2\mu_j'\Sig_0^{-1}U'_0x_i) \right\} \prod_{\ell = 1}^c \nu_{j\ell}^{I[y_i=\ell]}   
\]
for $ j = 1, \ldots, \infty$.   Once again, one may introduce slice sampling latent variables and implement the exact block Gibbs sampler or use the block Gibbs sampler directly to make the total number of states finite.  

\item[{\bf Step 9}.] 
Update the $\bm{\nu}_j$'s by sampling from $[\bm{\nu}_j | - ] \sim Dir(a^*_1, \ldots, a^*_c)$, where $a^*_{\ell} = \sum_{i=1}^nI[y_i = \ell, S_i=j] + a_{\ell}$ for $\ell = 1, \ldots, c$.
\end{enumerate}


\subsection{Simulation Study}
To demonstrate the performance of the classifier we conduct  a small simulation study.  Synthetic data sets are generated using two methods.  The first method treats the PSC as a data generating mechanism, the second is similar to the data generating scheme found on page 16 of Hastie, Tibshirani and Freedman \cite{HTF:2008} (here after referred to as HTF). We briefly describe both.  

When the PSC is being used as a data generating mechanism, the $X$ matrix is generated using \eqref{e61}.  We set $m=100$, $\sigma^2=0.1$, and $k=2$.  As this produces a feature space with three clusters,  $Y$ takes on values in $\{1, 2, 3\}$ with probabilities $[\til w_1(U'_0X), \til w_2(U'_0X), \til w_3(U'_0X)]$ where $U_0$ is found in \eqref{e62}.   The second data generating scenario consists of two classes with 100 observations each. The observations are drawn from the Gaussian mixture $\sum_{j=1}^{10} 1/10 N_{100}(m_j,1/5I)$.  The 10 means, $m_j$,  for the two classes are generated independently from $N_{100}(\eta_1, I)$ and $N_{100}(\eta_2, I)$ respectively  ($\eta_1$ and $\eta_2$ are defined in \eqref{e61}).
For each scenario 100 data sets are generated.  For the first,  100 training and 100 testing observations were generated and for the second 200 test and 200 training observations were used.   The PSC, $k$ nearest neighbor (KNN), and mixture discriminant analysis (MDA) were employed to classify the response from the testing data sets.  KNN and MDA procedures were selected as competitors because KNN is an algorithmic based procedure that is known to perform well in a variety of settings (see HTF) and MDA is a flexible model based Gaussian mixture classifier (see  Hastie and Tibshirani \cite{hastie96}).  We employ the {\tt knn} \cite{ClassPack} and {\tt mda} \cite{mdaPack} functions both of which are available freely from the {\tt R} software  \cite{Rsoft} to implement the KNN and MDA methods.  For the KNN we set $k=6$ for data generated from the PSC and $k=25$ for HTF data. These values were deemed to produce the smallest misclassification rate for a few synthetic data sets from both data generating scenarios.   For the same reason, with regards to the MDA, the number of components for each classes Gaussian mixture was set at 5.  Choosing $k$ in this manner provides an advantage to KNN and MDA when comparing misclassification rates to  the PSC.  

For the PSC, 1000 MCMC iterates were collected after a burn-in of 10,000 and thinning of 100.  Convergence was assessed using history plots of the MCMC draws for a few data sets.  The out of sample misclassification rates averaged over the 100 data sets can be found under each procedures respective heading in Table \ref{SimStudyResults2}.

\begin{table}[htdp]
\caption{Misclassification rates from the simulation study.  Data were generated using the PSC and the method detailed on page 16 of Hastie, Tibhshirani and Feedman (HTF)\protect \cite{HTF:2008} }
\begin{center}
\begin{tabular}{c cccc}
Data Generating & & & &\\
Mechanism &   PSC  & KNN  & MDA\\ \midrule 
PSC & 0.060 & 0.158  & 0.639 \\ 
HTF & 0.047 &  0.269  & 0.369\\ \bottomrule
\end{tabular}
\end{center}
\label{SimStudyResults2}
\end{table}%

It appears as if the PSC is able to more accurately classify the categorical response from the testing data compared to KNN and MDA.  This appears to be true regardless of what $k$ is fixed to be.  Preliminary studies  indicated that the PSC classifier still out preformed KNN and MDA (though not as drastically) even with correlated and non-Gaussian non-signal predictors.

\subsection{Illustration on Real Datasets}  
We now apply the PSC to two real data sets both of which are readily available in {\tt R}.  The first consists of two classes and 7 quantitative predictors.    The predictors are physiological measurements taken on Pima Indian women with the goal of predicting the presence or absence of diabetes.  To these 7 predictors we add another 93 which are comprised of random standard Gaussian draws.   The dataset is split randomly into training and testing sections.  The training section consists of 200 women, 68 of which are diagnosed with diabetes, while the testing section consists of 332 women, 109 of which are diagnosed with diabetes.   

The second data set we consider is the so called iris data set.  Here the response consists of three classes each one representing a specific flower species.  The four predictors are length and width measurements corresponding to the sepal and petal of a flower.  The goal is to use these four measurements to predict the flower species.  To the four predictors we add 96 that are comprised of random standard Gaussian draws.   The data set consists of 150 observations with each flower species having 50.  Fifty observations were randomly selected to comprise the testing data while the remaining 100 were used for the training data set.

To both data sets we applied the PSC in addition to KNN classifier and a MDA classifier.  For the KNN classifier, we chose the value of $k$ that minimized the misclassification rate which turned out to be $k=5$ for the iris data and $k=24$ for the diabetes data.  Similarly, the number of components comprising the  Gaussian mixtures of the MDA classifier  was selected on the basis of minimizing the misclassification rate.  The number of components turned out be 5 for the iris data and 7 for the diabetes data.  Note that choosing $k$ in this manner gives an unfair advantage to KNN and MDA relative to PSC, which does not use the test data at all in training.  We fit the PSC to both data sets by collecting 1000 MCMC iterates after a burn-in of 10,000 and thinning of 100.  Convergence was monitored using trace plots from two chains that were started at different values.   Prior to analysis variables were standardized.  The misclassification rates can be found in Table \ref{RealDataResults}

\begin{table}[htdp]
\caption{Misclassification rates for the iris and diabetes data sets. }
\begin{center}
\begin{tabular}{c cccccc}
Data set &   PSC &  KNN  & MDA\\ \midrule 
Iris & 		0.22 &  0.55  & 0.51 \\ 
Diabetes & 	0.26 &  0.29  & 0.37 \\ \bottomrule
\end{tabular}
\end{center}
\label{RealDataResults}
\end{table}%

It appears that the PSC was able to classify the testing data response in the presence of a high dimensional feature space much more accurately than either KNN or MDA.

\section{Conclusions}
This article has proposed a novel methodology for nonparametric Bayesian learning of an affine subspace underlying high-dimensional data.  Clearly, massive-dimensional data are now commonplace and there is a need for flexible methods for dimensionality reduction that avoid parametric assumptions.  In this context, the Bayesian paradigm has substantial advantages over commonly used machine learning, computer science and frequentist statistical methods that obtain a point estimate of the subspace or manifold which the data are concentrated near.  As there is unavoidably substantial uncertainty in subspace or manifold learning, it is important to fully account for this uncertainty to avoid misleading inferences and obtain appropriate measures of uncertainty in estimating densities, performing predictions and identifying important predictors.  We accomplish this in a Bayesian manner by placing a probability model over the space of affine subspaces, while developing a simple and efficient computational algorithm relying on Gibbs sampling to estimate the subspace and its dimension or model-average over subspaces of different dimension. The model is theoretically proved to be highly flexible and posterior consistency is achieved under appropriate prior choices. 
The proposed model and computational algorithm should be broadly useful beyond the density estimation and classification settings we have considered.

A potential alternative to our approach mentioned in Section 1 is to use a mixture of sparse factor models to build a tangent space approximation to the manifold the data are concentrated near.  Sparse Bayesian normal linear factor models are a successful approach for dimensionality reduction (Carvalho \emph{et al}., \cite{carvalho}; Bhattacharya and Dunson \cite{abdd}), but make restrictive normality assumptions and are limited in their ability to reduce dimensionality by linearity assumptions.  By mixing factor models, one can certainly obtain a more flexible characterization, but challenging computational issues arise in accommodating uncertainty in the number of factors and locations of zeros in the factor loadings matrix for each of the multivariate Gaussian components in the mixtures.  Indeed, even in modest dimensions for a normal linear factor models, Lopes and West \cite{lopes} encountered difficulties in efficiently inferring the number of factors, and recommending using a reversible jump MCMC algorithm that required a preliminary MCMC run for each choice of the number of factors.  For mixture of factor models, one obtains a extremely rich over-parametrized black box.  We propose a fundamentally new alternative that directly specifies an identifiable model based on geometry, while also developing an efficient Gibbs sampler that can infer the dimension of the subspace automatically without RJMCMC.  Although our initial focus was on data in a Euclidean space, related models can be developed for non-Euclidean manifold data, as we will explore in ongoing work.

{\bf Acknowlegements}:  This work was partially supported by Award Number R01ES017436 from the National Institute of Environmental Health Sciences.  The content is solely the responsibility of the authors and does not necessarily represent the official views of the National Institute of Environmental Health Sciences or the National Institutes of Health.

\appendix
\appendixpage
\section{Proofs}\label{a1}

As a reminder in what follows  $B_{r,m}$ refers to the set $\{ x \in \Re^m \colon \|x\| \le r \}$. For a subset $\mc{D}$ of densities and $\ep>0$, the $L_1$-metric entropy $N(\ep,\mc{D})$ is defined as the logarithm of the minimum number of $\epsilon$-sized (or smaller) $L_1$ subsets needed to cover $\mc{D}$.

\subsection{Proof of Lemma \eqref{l1}}
\begin{proof}
 Any density $f$ in $\mc{D}_n^\ep$ can be expressed as 
$\int_{\Re^m} N_m(\nu,\Sig) Q(d\nu)$ with $\Sig = U_0 \Sig_0 U_0' + \sig_0^2(I_m - U_0 U_0')$, $Q = P \circ \phi^{-1}$, 
$\phi(x) = U_0x$, and $(k,U_0,\th,\Sig_0,\sig,P)\in H_n^\ep$. 
The assumption on $\pi_2$ and $H_n^\ep$ will imply that $\Sig$ has all its eigen-values in $[h_n^2,A^2]$.

We also claim that 
$Q(B_{\s{2}r_n,m}^c) < \ep$.
To see that, note that
$\|\phi(\mu)\|^2 = \|\mu\|^2 + \|\th\|^2 \le 2r_n^2$ whenever $\|\mu\|\le r_n$ and $\|\th\|\le r_n$.
Hence $B_{r_n,k} \subseteq \phi^{-1}(B_{\s{2}r_n,m})$ if $\|\th\| \le r_n$.
Therefore $\ep > P(B_{r_n,k}^c) \ge P\big( (\phi^{-1}(B_{\s{2}r_n,m}))^c \big) = P\circ\phi^{-1}\big(B_{\s{2}r_n,m}^c \big)$ 
for all $(P,\th) \in H_n^\ep$. Hence the claim follows.

Therefore $$ \mc{D}_n^\ep \subseteq \til{\mc{D}}_n^\ep = \{ f = \int N_m(\nu,\Sig) Q(d\nu): Q(B_{\s{2}r_n,m}^c) < \ep, \ \l(\Sig) \in [h_n^2,A^2] \}, $$
$\l(\Sig)$ denoting the eigen-values of $\Sig$. From Lemma 1 of Wu and Ghosal \cite{WuGhosal2010}, it follows that
$N(\ep,\til{\mc{D}}_n^\ep) \le C (r_n/h_n)^m$ and this completes the proof.
\end{proof}

\subsection{Proof of Lemma \eqref{l2}}
The proof is similar in scope to the proof of Lemma 2  in Wu and Ghosal \cite{WuGhosal2010}. Throughout the proof, $C$ will denote constant independent of $n$.

\begin{proof}
Given $k, U, \th, \bm{\sig}$ and $\un{\mu}_n=$ $\mu_1, \ldots, \mu_n$ iid $P$, $X_i \sim N_m\big( \phi(\mu_i), \Sig\big)$, $i=1,\ldots, n$, independently and are independent of $P$.
Hence
$$
Pr\big( P(B_{r_n,k}^c) \ge \ep \big| k,\mv{X}_n \big) = E\big( Pr\big( P(B_{r_n,k}^c) \ge \ep \big| k,\un{\mu}_n \big) \big| k,\mv{X}_n \big).
$$  
From \cite{ferguson}, given $\un{\mu}_n$ and $k$, for $A\subseteq \Re^k$,
$P(A) \sim Beta\big(w_kP_k(A) + N(A), w_k(1-P_k) + n - N(A)  \big)$
where $N(A) = \sum_{i=1}^n I_{\{\mu_i \in A\}}$.
Hence using the Markov inequality,
$$
Pr\big( P(B_{r_n,k}^c) \ge \ep \big| k,\un{\mu}_n \big) \le \frac{w_k P_k(B_{r_n,k}^c) + N(B_{r_n,k}^c)}{\ep(n+w_k)}.
$$
Therefore
\begin{align*}
E\big( Pr\big( P(B_{r_n,k}^c) \ge \ep \big| k, \mv{X}_n \big) \le \frac{w_k P_k(B_{r_n,k}^c)}{\ep(n+w_k)} + \frac{1}{\ep(n+w_k)} \sum_{i=1}^n Pr\big( \mu_i \in B_{r_n,k}^c \big| k, \mv{X}_n \big).
\end{align*}
Denote the above two terms as $T_1$ and $T_2$. Then $E_{f_t}T_1 = T_1 \lra 0$ as $r_n\ra \infty$.
Under the marginal prior given $k$, $\un{\mu}_n$ has an exchangable distribution $\pi_n(\un{\mu}_n|k)$ on $(\Re^k)^n$ (see \cite{ferguson}). Also since $\mv{X}_n$ are iid given $f_t$, it follows that
$$
E_{f_t}(T_2) = \frac{n}{\ep(n+w_k)} E_{f_t} \big\{ Pr\big( \mu_1 \in B_{r_n,k}^c \big| k, \mv{X}_n \big) \big\}.
$$
Now
\begin{align*}
Pr\big( \mu_1 \in B_{r_n,k}^c \big| k, \mv{X}_n \big) \le  Pr\big( \mu_1 \in B_{r_n,k}^c, \min(\bm{\sig}) > h_n \big| k, \mv{X}_n \big) + \\
Pr( \min(\bm{\sig}) \le h_n \big| k, \mv{X}_n ).
\end{align*}
The last term above converges to $0$ a.s. by the assumption  on $\pi_2$.
Hence to complete the proof, it remains to show that
$$
E_{f_t} \big\{ Pr\big( \mu_1 \in B_{r_n,k}^c, \min(\bm{\sig}) > h_n \big| k, \mv{X}_n \big) \big\} \lra 0 \t{ as } n\ra \infty.
$$
To compute the probability in above, we denote by $\pi_{1n}(\mu_1|\mu_{-1},k)$ the conditional distribution of $\mu_1$ given $\mu_{-1} = (\mu_2,\ldots,\mu_n)$ , and by $\pi_{-1n}(\mu_{-1}|k)$ the marginal
distribution of $\mu_{-1}$ under the joint $\pi_n$. 
Then
$$
Pr\big( \mu_1 \in B_{r_n,k}^c, \min(\bm{\sig}) > h_n \big| k, \mv{X}_n \big) = A(\mv{X}_n)/B(\mv{X}_n)
$$
where $A(\mv{X}_n) =$
\begin{align*}
 \mathop{\int}_{\min(\bm{\sig})>h_n,\|\mu_1\|>r_n} \prod_{i=1}^n N_m(X_i; \phi(\mu), \Sig) d\pi_{1n}(\mu_1|\mu_{-1},k) d\pi_{-1n}(\mu_{-1}|k)d\pi_1(U_0,\th|k)d\pi_2(\bm{\sig}|k)
\end{align*}
and $B(\mv{X}_n) =$
$$
\int \prod_{i=1}^n  N_m(X_i; \phi(\mu), \Sig) d\pi_{1n}(\mu_1|\mu_{-1},k) d\pi_{-1n}(\mu_{-1}|k)d\pi_1(U_0,\th|k)d\pi_2(\bm{\sig}|k).
$$
We use $E_{f_t}\{A(\mv{X}_n)/B(\mv{X}_n)\} \le$
\begin{align}\label{e11}
\mathop{\sup}_{X_1 \in B_{r_n/2,m}}\frac{A(\mv{X}_n)}{B(\mv{X}_n)} \int_{B_{r_n/2,m}}f_t(x)dx +  \int_{ B_{r_n/2,m}^c} f_t(x)dx.
\end{align}
and upper bound the terms in above.

First we upper bound $A(\mv{X}_n)$ when $\|X_1\| \le r_n/2$.
We express $N_m(X_1; \phi(\mu_1), \Sig)$ as 
$$
N_k(U_0'X_1;\mu_1, \Sig_0)  
$$
and note that $\|X_1\|\le r_n/2$, $\|\mu_1\| > r_n$ and $h_n < \sig_j \le A $ $\forall j \le k$ implies
$$
N_k(U_0'X_1;\mu_1, \Sig_0) \le C h_n^{-k} \exp \frac{-r_n^2}{8A^2}.
$$
Therefore $A(\mv{X}_n) \le$
\begin{align} \label{e9}
\begin{split}
 C h_n^{-k} \exp\frac{-r_n^2}{8A^2} \int (\sig^{-2})^{\frac{m-k}{2}} \exp\frac{-1}{2\sig^2}(X_1 - \th)'(I_m - U_0U_0')(X_1 - \th)\\
 \prod_{i=2}^n N_m(X_i;\phi(\mu_i),\Sig) 
d\pi_{-1n}(\mu_{-1}|k)d\pi_1(U_0,\th|k)d\pi_2(\bm{\sig}|k).
\end{split}
\end{align}

Next we lower bound $B(\mv{X}_n)$ when $X_1 \in B_{r_n/2,m}$.
The conditional distribution $\pi_{1n}$ can be expressed as 
$
\frac{1}{w_k + n-1}\sum_{i=2}^n \del_{\mu_i} + \frac{w_k}{w_k + n-1} P_k
$ (see \cite{ferguson}).
Hence $B(\mv{X}_n) \ge$
$$
\frac{w_k}{w_k + n-1} \int \prod_{i=1}^n N_m(X_i; \phi(\mu_i), \Sig) p_k(\mu_1)d\mu_1 d\pi_{-1n}(\mu_{-1}|k)d\pi_1(U,\th|k)d\pi_2(\bm{\sig}|k).
$$
Now 
$$
\int N_k( U_0'X_1;\mu_1, \Sig_0) p_k(\mu_1)d\mu_1 \ge 
\int_S N_k( U_0'X_1;\mu_1, \Sig_0)  p_k(\mu_1)d\mu_1
$$
where 
$$
S = \{\mu_1: \sum_{l=1}^k \sig_l^2 (U_k'X_1 - \mu_1)^2_l \le 1 \}.
$$
For $\mu_1 \in S$, $N_k\big( U_0'X_1;\mu_1, \Sig_0) \ge \prod_1^k \sig_j^{-1} e^{-1/2}$
and $p_k(\mu_1) \ge \del_{kn}$ with $\del_{kn}$ defined in the Lemma.
Therefore
$$
\int_S N_k( U_0'X_1;\mu_1, \Sig_0)  p_k(\mu_1)d\mu_1 \ge  C\del_{kn}\prod_1^k \sig_j^{-1}\int_S d\mu_1 = C\del_{kn}
$$
and hence when $\|X_1\| \le r_n/2$, $B(\mv{X}_n) \ge$
\begin{align} \label{e10}
\begin{split}
 C n^{-1} \del_{kn} \int (\sig^{-2})^{\frac{m-k}{2}} \exp\frac{-1}{2\sig^2}(X_1 - \th)'(I_m - U_0U_0')(X_1 - \th) \prod_{i=2}^n N_m(X_i; \phi(\mu_i), \Sig) \\
d\pi_{-1n}(\mu_{-1}|k) d\pi_1(U_0,\th|k) d\pi_2(\bm{\sig}|k).
\end{split}
\end{align}
Combining \eqref{e9} and \eqref{e10}, we get
$$
\sup_{\|X_1\| \le r_n/2} \frac{A(\mv{X}_n)}{B(\mv{X}_n)} \le C n \del_{kn}^{-1} h_n^{-k} \exp(-r_n^2/8A^2).
$$
Plug this in \eqref{e11} to conclude
$E_{f_t}\{A(\mv{X}_n)/B(\mv{X}_n)\} \le$
\begin{align} \label{e15}
C n \del_{kn}^{-1} h_n^{-k} \exp(-r_n^2/8A^2) + Pr_{f_t}(\|X\| > r_n/2)
\end{align}
which converges to zero by assumption.

Under assumption {\bf B1'} and $\sum r_n^{-2(1+\alp)m} < \infty$ the sequence in \eqref{e15} has a finite sum which results in the stronger conclusion.
This completes the proof.
\end{proof}

\subsection{Proof of Corollary \eqref{t5}}
\begin{proof}
By Theorem~\ref{t4}, to show a.s. strong posterior consistency, we need to get positive sequences $r_n$ and $h_n$ which satisfy
\begin{align}
 n^{-1}(r_n/h_n)^m \lra 0, \ \sum r_n^{-2(1+\alp)m} < \infty, \t{ and} \label{e21}\\
\sum_{n=1}^\infty n\del_{kn}^{-1}h_n^{-k}\exp(-r_n^2/8A^2) < \infty, \label{e22}
\end{align}
and the prior probabilities $Pr(\|\th\|> r_n |k)$ and $Pr(\min(\bm{\sig}) < h_n |k)$ decay exponentially.
Set $r_n = n^{1/a}$ and $h_n = n^{-1/b}$.
Then \eqref{e21} is clearly satisfied.

By the choice of $p_k$, $k\ge 1$, it is easy to check that
$\del_{kn} \ge C \exp\frac{-r_n^2}{2\tau_k^2}$ with $C$ denoting positive constants independent of $n$ all throughout.
Then \eqref{e22} is clearly satisfied because of the assumption $\tau_k^2 > 4A^2$.

Because $\|\th\|^a$ follows a Gamma distribution given $k$, $k \le m-1$, the probability $Pr(\|\th\|> r_n |k)$ can be upper bounded by $C\exp(-\l r_n^a)$ for some $\l>0$. This decays exponentially 
with $r_n = n^{1/a}$.

Lastly it remains to check that $Pr(\min(\bm{\sig}) < h_n |k)$, decays exponentially. When the coordinates of $\bm{\sig}$ are all equal, the probability can be upper bounded by 
$C\exp(-\l h_n^{-b})$ for some $\l>0$. This decays exponentially with $h_n = n^{-1/b}$.
In case the coordinates are iid, the probability can be upper bounded by 
$Cn\exp(-\l h_n^{-b})$ which also decays exponentially by the choice of $h_n$. 
\end{proof}

\subsection{Proof of Theorem \eqref{t2}}

\begin{proof}
Simplify $f_1$ as
\begin{align} \label{e19}
&f_1(R,\th) = f_1(\bar R, \bar\th) + \|R - \bar R\|^2 + \|\th - \bar\th\|^2  \notag \\
&= f_1(\bar R, \bar\th) + \|R - \bar R\|^2 + \|R\bar\th\|^2 + \|(I-R)(\th - \bar\th)\|^2 \notag \\
&\ge f_1(\bar R, \bar\th) + \|R - \bar R\|^2 + \|R\bar\th\|^2.
\end{align}
Equality holds in \eqref{e19} iff $\th = (I-R)\bar\th$. Then
$$
f_1(R,\th) = k - \mr{Tr}\{(2\bar R - \bar\th\bar\th')R\} + C
$$
where $k = $Rank($R$) and $C$ denotes something not depending on $R,\th$.
From the proof of Proposition 11.1\cite{bhatta}, given $k$ one can show that the value of $R$ minimizing $f_1$ above is $\sum_{j=1}^k U_j U_j'$ and the minimizer is unique iff $\l_k > \l_{k+1}$. Then 
$$
f_1(R,\th) = k - \sum_{j=1}^k\l_j + C.
$$
Now one needs to find the $k$ minimizing the above risk which is as mentioned. This completes the proof.
\end{proof}

\subsection{Proof of Theorem \eqref{t3}}

\begin{proof}
The minimizer $w=\bar w$ is obvious. Then
$$
f_2(U,\bar w) = \|U - \bar U\|^2 + C = k_1 - 2\mr{Tr}\bar U_{(k_1)}'U_{(k_1)} + C,
$$
$k_1$ being the rank of $U$ and $C$ symbolizing any constant not depending on $U$.
For $k_1$ fixed, it is proved in Theorem 10.2\cite{bhatta} that  the minimizer $U$ is as in the theorem. It is unique iff $\bar U_{(k_1)}'\bar U_{(k_1)}$ is invertible. Plug that $U$ and the risk function becomes, as a function of $k_1$, 
$$
f_3(k_1) = k_1 - 2\mr{Tr}(\bar U_{(k_1)}'\bar U_{(k_1)})^{1/2}.
$$
We find the value of $k_1$ between $1$ and $m$ minimizing $f_3$ and set $k=k_1-1$.
\end{proof}

\singlespace

\bibliographystyle{plain} 
\bibliography{reference}

\end{document}